\documentclass[11pt]{article}
\usepackage{tikz}
\usepackage{fullpage}
\usepackage{amsmath,amsthm,epsfig,amssymb}
\usepackage{xcolor}

\newcommand{\kP}{$k$-class}
\newcommand{\kPs}{$k$-classes}
\newcommand{\cM}{\mathcal{M}}
\newcommand{\psize}{N^*}

\newcommand{\R}{\mathbb{R}}
\newcommand{\cP}{\mathcal{P}}

\newcommand{\Gap}{\operatorname{Gap}}

\newcommand{\Par}[1]{#1^{\|}}

\renewcommand{\a}{x_{12}}
\renewcommand{\b}{x_{11}}
\renewcommand{\c}{x_{22}}
\renewcommand{\d}{x_{21}}
\newcommand{\A}{A}
\newcommand{\C}{C}
\newcommand{\transmat}{P}
\newcommand{\Phat}{\hat{\transmat}}
\newcommand{\Ptilde}{\tilde{\transmat}}
\newcommand{\Pbar}{\bar{\transmat}}
\newcommand{\Ahat}{\hat{\A}}
\newcommand{\Atilde}{\tilde{\A}}
\newcommand{\Abar}{\bar{\A}}
\newcommand{\pihat}{\hat{\pi}}
\newcommand{\pitilde}{\tilde{\pi}}
\newcommand{\pibar}{\bar{\pi}}
\newcommand{\zhat}{\hat{Z}}
\newcommand{\ztilde}{\tilde{Z}}

\newcommand{\xhat}{\hat{x}}
\newcommand{\xtilde}{\tilde{x}}

\newcommand{\lambdatilde}{\mu}
\newcommand{\lambdabar}{\bar{\lambda}}
\newcommand{\deltahpeh}{\hat{\delta}_1}
\newcommand{\deltahpet}{\hat{\delta}_2}%
\newcommand{\deltahpah}{\hat{\delta}_3}%
\newcommand{\deltahpat}{\hat{\delta}_4}%
\newcommand{\deltatpet}{\tilde{\delta}_1}%
\newcommand{\gammahat}{\hat{\gamma}}
\newcommand{\gammatilde}{\tilde{\gamma}}
\newcommand{\gammabar}{\bar{\gamma}}
\newcommand{\eigenvec}{v}

\newcommand{\vtilde}{\eigenvec}
\newcommand{\numrests}{r}
\newcommand{\rhat}{\hat{\numrests}}
\newcommand{\rtilde}{\tilde{\numrests}}
\newcommand{\Omegahat}{\hat{\Omega}}
\newcommand{\Omegatilde}{\tilde{\Omega}}
\newcommand{\elemone}{\sigma}
\newcommand{\elemtwo}{\tau}
\newcommand{\bigx}{z}
\newcommand{\m}{\cM}

\newcommand{\gmin}{\gminhat}
\newcommand{\gminhat}{\gammahat_{\min}}
\newcommand{\gmintilde}{\gammatilde_{\min}}
\newcommand{\Perp}[1]{#1^{\perp}}

\newcommand{\p}{\mathcal{P}}
\newcommand{\tk}{\tau_{\text{T}^*}} 
\newcommand{\ttk}{\tau_{\text{T}}}
\newcommand{\mn}{\cM_{\text{n}}}

\newcommand{\mk}{\cM_{\text{p}}} 
\newcommand{\mtk}{\cM_{\text{T}}}
\newcommand{\me}{\cM_{\text{ex}}}

\newcommand{\comment}[1]{}

\usepackage{ifthen} 
\newboolean{includefigs} 
\setboolean{includefigs}{true} 
\newcommand{\condcomment}[2]{\ifthenelse{#1}{#2}{}}
\theoremstyle{plain}
\newtheorem{theorem}{Theorem}
\newtheorem{Proposition}{Proposition}
\newtheorem{Definition}{Definition}
\newtheorem{Lemma}[theorem]{Lemma}
\newtheorem{Corollary}[theorem]{Corollary}

\newtheorem{Remark}[theorem]{Remark}

\newtheorem{innercustomthm}{Theorem}
\newenvironment{customthm}[1]{\renewcommand\theinnercustomthm{#1}\innercustomthm}{\endinnercustomthm}
\title{Iterated Decomposition of Biased Permutations Via New Bounds on the Spectral Gap of Markov chains}

\author{
  Sarah Miracle
    \thanks{Computer and Information Sciences, University of St.\ Thomas, St.\ Paul, MN 55105; {\tt sarah.miracle@stthomas.edu}.}
 \and Amanda Pascoe Streib
   \thanks{Center for Computing Sciences, Bowie, MD 20715-4300; {\tt
       ampasco@super.org}.}
   \and Noah Streib
      \thanks{Center for Computing Sciences, Bowie, MD 20715-4300; {\tt
       nsstrei@super.org}.}
}

\begin{document}
\date{}
\maketitle

{\color{red}


}
\thispagestyle{empty}
\begin{abstract}

The spectral gap of a Markov chain can be bounded by the spectral gaps of constituent ``restriction" chains and a ``projection" chain, and the strength of such a bound is the content of various decomposition theorems.  In this paper, we introduce a new parameter that allows us to improve upon these bounds.  We further define a notion of orthogonality between the restriction chains and ``complementary'' restriction chains.  This leads to a new Complementary Decomposition theorem, which does not require analyzing the projection chain.  For $\epsilon$-orthogonal chains, this theorem may be iterated $O(1/\epsilon)$ times while only giving away a constant multiplicative factor on the overall spectral gap.
As an application, we provide a $1/n$-orthogonal decomposition of the nearest neighbor Markov chain over $k$-class biased monotone permutations on $[n]$, as long as the number of particles in each class is at least $C \log n$.   This allows us to apply the Complementary Decomposition theorem iteratively $n$ times to prove the first polynomial bound on the spectral gap when $k$ is as large as $\Theta(n/\log n)$. The previous best known bound assumed $k$ was at most a constant.


\end{abstract}

\newpage
\setcounter{page}{1}







\section{Introduction} 

The decomposition method for Markov chains allows one to bound the 
spectral gap of a Markov chain in terms of the spectral gaps of 
constituent (hopefully simpler) Markov chains.  The method was first 
introduced by Madras and Randall~\cite{madr}, and has been subsequently 
used and modified to produce the first polynomial time bounds on the 
spectral gaps of many interesting Markov chains~\cite{CDF00,Dest03,Ding2010, EMT18, GHS08,HW16, 
jstv04,mh15,MR00,mrs,PS17, R01}. In 
this paper, we consider a disjoint decomposition in the style 
of~\cite{MR00}.  We assume throughout the paper that $\m$ is a finite, ergodic Markov chain that is reversible with respect to the distribution $\pi$.
Suppose $\Omega=\cup_{i=1}^{\rhat}\Omegahat_i$ is a 
partition of the state space of a Markov chain $\m$, and let 
$\gammahat_i$ be the spectral gap of $\m$ restricted to $\Omegahat_i$.  
The decomposition theorem of~\cite{MR00} states that the spectral gap 
$\gamma$ of $\m$ satisfies $\gamma \geq \frac{1}{2}\gmin\gammabar$, 
where $\gmin=\min_i \gammahat_i$ and $\gammabar$ is the spectral gap of 
a certain \emph{projection} chain over states $\{1,2,\ldots, \rhat\}$. 


There has been 
significant effort towards improving the decomposition technique by 
providing stronger bounds in special 
cases~\cite{Dest03,EMT18,GHS08,jstv04,mh15,MR00,PS17,R01}.
While $\gamma$ may indeed be on the order of 
$\gmin\gammabar$---one example is the random walk on the 
path graph of length $n$, decomposed into two smaller path graphs---there are instances in which it may instead satisfy the much 
larger bound 
$\gamma \geq c \min\{\gmin, \gammabar\}, $
for some constant $c$. The simplest such example is
the direct product of two 
independent Markov chains~\cite{bmrs,EMT18}; in this case, $c=1$. 
Jerrum, Son, Tetali, and Vigoda~\cite{jstv04} considered two
related quantities: the Poincar\'{e} and log-Sobolev constants.  They 
introduced 
a parameter $T=\max_i\max_{\elemone\in\Omegahat_i}\sum_{\elemtwo\in 
\Omega\setminus\Omegahat_i}\transmat(\elemone,\elemtwo)$, which can be 
seen as the maximum probability of escape from one part of the partition in a 
single step of $\transmat$. 
They produced a bound on the order of the minimum gap 
when $T$ is on the order of $\gammabar$.  They also provided improved
bounds when another parameter $\eta$ is close to zero; this requires a pointwise regularity condition.
Destainville~\cite{Dest03} introduced a ``multi-decomposition" scheme that uses $m$ different partitions of $\Omega$. 
The bound obtained depends on the norm of a ``multi-projection" operator $\Pi$.  
More recently, Pillai and Smith~\cite{PS17} introduced other
conditions in order to directly bound the mixing time by a constant times the maximum of the mixing times of the projection and the restrictions.


Tight bounds are especially important when applying the decomposition method 
iteratively.  For example, we consider as an application a certain transposition chain $\mk$ over biased permutations in $S_n$, where there is a natural way 
to decompose the space iteratively $n$ times.  At each level of the induction, $\gammabar=\Theta(n^{-2})$, so the 
original bound of~\cite{MR00} yields $\gamma=\Omega(n^{-2n})$ for the final iteration.  
Even a bound of the form $\gamma \geq c \min\{\gmin, \gammabar\}$ such as the one in~\cite{PS17} introduces a factor of $c$ for each application, and  
yields a bound that is an inverse exponential in $n$.
The bounds in~\cite{jstv04} are iterable in some cases, but $\mk$ does not satisfy those conditions; in fact, the parameter $\eta$ is exponential in $n$ in this case.
Destainville's bounds~\cite{Dest03} are iterable as well, but can be challenging to apply. 

In this paper, we 
produce a set of techniques that allow us to get iterable bounds for a
more general class of decomposable Markov chains.    
Our first decomposition theorem requires a new parameter $\deltahpet$
(defined in Section~\ref{sec:decomp}); the function below is minimized when $\deltahpet$ is minimized.


\begin{theorem}\label{thm:main}
  Let $\rho=\sqrt{(1-\deltahpet)/\gammabar}.$ Then   $\Gap(\m)\geq \displaystyle\min_{p^2+q^2=1}
\gmin q^2+\gammabar\left(q\rho - p\right)^2.$
\end{theorem}
\noindent
In many cases this bound already improves upon what is known.
 When $\deltahpet=1$, we get
$\Gap(\m)=\min\{\gmin, \gammabar\}$. 
If 
$\m$ is lazy then $\deltahpet\geq 0$ and we get $\gamma\geq \gmin\gammabar/3$ (see
Section~\ref{sec:compare}).  
Moreover, we will show $\deltahpet \geq 1-2T$, so
Theorem~\ref{thm:main} can be seen as a generalization of Theorem~1
of~\cite{jstv04}, except that it instead bounds the spectral gap.  In fact, in Corollary~\ref{cor:compareT} we reprove that result. 
In particular, if $T/\gammabar$ is a constant, then
we get within a constant of the minimum gap as well.

On the way toward proving Theorem~\ref{thm:main}, we derive Theorem~\ref{thm:primaldual}, which achieves a tight bound on $\gamma$ 
without appealing to a bound on the projection chain.  Let $\transmat$ 
be the transition matrix of $\m$, and define the \emph{complementary} 
restrictions $ \Ptilde_1, \Ptilde_2, \ldots,\Ptilde_{\rtilde}$
to contain the transitions of $\transmat$ that are not in any restriction 
$\Phat_{i}$, and further define 
$\gmintilde$ analogously to $\gminhat$\footnote{If some restriction or complementary restriction has a single element, its spectral gap is taken to be $1$.}.  
We will define $\Perp{\xhat}$ and $\Perp{\xtilde}$ to be orthogonal projections of a vector $x$ onto the complement of the eigenspace of the top eigenvectors of certain matrices containing the $P_i$'s and $\Ptilde_j$'s, respectively (see Section~\ref{sec:tight}). Then we prove:
\begin{theorem}\label{thm:primaldual}
$\Gap(\m)\geq \min_{x\perp \sqrt{\pi}, \|x\|=1}  
\gminhat\|\Perp{\xhat}\|^2+\gmintilde\|\Perp{\xtilde}\|^2.$
\end{theorem}

\noindent This theorem turns out to be very similar to a special case of the main result in~\cite{Dest03}, where $\|\Perp{\xhat}\|^2+\|\Perp{\xtilde}\|^2$ is replaced by a function of the norm of $\Pi$.  Bounding these norms is essential to making these results useful, since unfortunately, 
the Markov chain $\m$ can require exponential time to mix even if all 
of the restrictions and complementary restrictions are polynomially 
mixing\footnote{Indeed, the introduction of the projection chain 
in~\cite{madr} was a key insight to the original decomposition theorem.}.

Unfortunately, bounding these norms can be challenging.  Destainville~\cite{Dest03} bounds the norm of the projection $\Pi$ by the spectral gap of a smaller matrix $\bar{\Pi}$.
In some cases, this gap can be analyzed directly, or even computationally for particular problem instances.  
We suspect one reason this result has not been applied more is that it is not immediately clear how to analyze this gap for more complex distributions.  Moreover, it is not clear under what circumstances this result can provide an advantage.  In this paper, we introduce a new technique for bounding the norms of these projections that appears to be more flexible and easier to analyze.  It is based on a notion of orthogonality between the top eigenvectors of the 
restrictions and complementary restrictions.  The intuition is that $\epsilon$-orthogonality implies that if a distribution is far from stationarity then it will either be far from stationarity on some restriction or on some complementary restriction.
This approach is particularly useful when the chain decomposes into pieces that are nearly independent.

\begin{theorem}\label{thm:indep}
If $\Phat_1, \Phat_2, \ldots,\Phat_{\rhat}$ and 
$\Ptilde_1, \Ptilde_2, \ldots, \Ptilde_{\rtilde}$
is an $\epsilon$-orthogonal decomposition of $\m$, then
\[ \Gap(\m) \geq \min\{\gminhat, \gmintilde\}\left(1 -\epsilon\right)^2.\]
\end{theorem}
\noindent This bound can be iterated $t$ times with 
only a constant overhead, as long as $\epsilon\leq 1/t.$  

As in~\cite{Dest03}, analysis of
$\epsilon$-orthogonality requires only a comparison between stationary
distributions and not an analysis of the dynamics. 
Define $\Omegatilde_j$ to be the state space of the Markov chain with transition matrix $\Ptilde_j$.
In order to bound $\epsilon$, we define a function
$r(i,j)=\pi(\Omegahat_i\cap\Omegatilde_j)/(\pi(\Omegahat_i)\pi(\Omegatilde_j))$
that indicates the relative independence of the restriction
$\Omegahat_i$ from $\Omegatilde_j$.  When $r(i,j)=1$ for all $i$ and $j$, then
$\epsilon=0$.  However, we do not require a strong pointwise bound on
$r(i,j)$, but instead a bound on its average value (see
Section~\ref{sec:indep} for details).  It is possible to prove
$\epsilon$ is very small even if $r(i,j)$ is far from 1 for
pathological pairs $i$ and $j$, as long as it is close to 1 for all but an inverse
polynomial fraction of the state space.  Importantly, this
holds even though the elements in this ``bad'' space are visited
polynomially often.  Indeed, for 
our application to permutations, where $r(i,j)$ can be
exponentially large or small for an inverse polynomial fraction of the
space, we still bound $\epsilon$ by $1/n$
and apply Theorem~\ref{thm:indep} iteratively $n$ times.  



\subsection{Techniques}

Our proofs are elementary and use only basic facts about eigenvalues 
and eigenvectors from linear algebra.  We have chosen to assume the 
Markov chains are discrete and finite to keep the proofs as accessible 
as possible.
We utilize the following standard characterization of the second 
largest eigenvalue $\lambda$ of a symmetric matrix $A$ with top 
eigenvector $\pi$: 
\begin{equation}\label{eq:eigen} 
\lambda = \max_{x\perp\pi}\frac{\langle x,xA\rangle}{\|x\|^2}
           = \max_{x\perp\pi:\|x\|=1}\langle x,xA\rangle.  
\end{equation} 
\noindent 
For a general reversible Markov chain with transition matrix $\transmat$, 
we apply Equation~\ref{eq:eigen} to a symmetric matrix $\A=\A(\transmat)$ that has the same eigenvalues as $\transmat$.

We apply the Vector Decomposition Method from
the expander graph literature (see, e.g.~\cite{RVW, expander-book}), and
decompose the vector $x$ into $\Perp{\xhat}+\Par{\xhat}$, where $\Par{\xhat}$ is
parallel to the top eigenvector of each restriction matrix.  The intuition of
this method is that if a particular distribution is far from
stationarity, then it will either be far from stationarity on some
part of the partition or on the projection, and therefore applying $\transmat$
brings us closer to stationarity. 
The benefit of this approach is that it allows us to quantify the independence of the restriction chains with the projection chain.  Using Equation~\ref{eq:eigen}, for any $x\perp \pi$, we need to bound 
\begin{equation}\label{eq:split}
\langle x, xA\rangle = \langle \Perp{\xhat},\Perp{\xhat}A\rangle + 
\langle\Par{\xhat},\Par{\xhat}A\rangle + 2\langle \Perp{\xhat},\Par{\xhat}A\rangle.
\end{equation}
It is easy to bound $\langle 
\Perp{\xhat},\Perp{\xhat}A\rangle$ and $\langle\Par{\xhat},\Par{\xhat}A\rangle$ using ideas from other decomposition results~\cite{jstv04, MR00}.  
The key to proving Theorem~\ref{thm:main} is our bound on $\langle \Perp{\xhat},\Par{\xhat}A\rangle$ in terms of $\deltahpet$.

\subsection{Application: Biased Permutations}
In Section~\ref{sec:perm}, we illustrate the power of this technique
by applying it to the biased permutation problem.
We are given a set of input probabilities ${\bf P} = \{p_{i,j}\}$ for all $1 \leq i, j
\leq n$ with $p_{i, j} = 1-p_{j, i}$.  At each step, the
nearest-neighbor transposition chain ${\mn}$
uniformly chooses a
pair of adjacent elements, $i$ and~$j$, and puts~$i$ ahead of~$j$ with
probability $p_{i,j}$, and~$j$ ahead of~$i$ with probability $p_{j,i}$. The Markov chain $\mn$ has been widely
studied~\cite{BBHM05,bmrs,DSC93b, diasha, wilson}; see,
e.g.~\cite{MS} for a review.  We say ${\bf P}$ is \emph{positively biased} if $p_{i,j}\geq
1/2$ for all $i<j$.  Without this condition, it is fairly straightforward to construct
parameter sets for which $\mn$ has exponential mixing time 
(see e.g.,~\cite{bmrs}).  Bhakta et al.~\cite{bmrs} showed that $\mn$ can require
exponential time to mix even for distributions with positive bias. 
Fill~\cite{F03b,F03a} introduced the following monotonicity
conditions:  $p_{i,j} \leq p_{i,j+1}$ and $p_{i,j} \geq p_{i+1,j}$ for
all $1 \leq i < j \leq n.$  Fill conjectured that $\mn$ is rapidly
mixing for all monotone, positively biased distributions and that the smallest spectral gap for $\mn$ 
is given by the uniform $p_{i,j}=1/2$ distribution.  
He confirmed these conjectures for $n\leq 3$ and gave experimental evidence for $n\leq 5$.

Bhakta et al.~\cite{bmrs} identified certain classes of ${\bf P}$
for which $\mn$ is actually a product of independent Markov chains.
Subsequently, two papers analyzed biased \kPs~\cite{HW16,MS}, where there are $k$ classes
of particles and particles from class $i$ and class $j$ interact
with the same probability $p_{i,j}$.  When $k=n$, this is the same as
the original permutation problem.  They also considered a more
powerful Markov chain $\mtk$ that can swap $i$ and $j$ if all elements
between them are smaller than both $i$ and $j$.  For this, it is sufficient to analyze the
$k$-particle process $\mk$, where elements within each particle class are in
fixed positions.  In general, these are not direct products of
independent Markov chains, and both of these papers used the disjoint
decomposition theorem of~\cite{MR00}.  They considered bounded \kPs, where $p_{i,j}/ p_{j,i}\leq q$ for all $i<j$ for some constant $q<1$.   In~\cite{HW16}, Haddadan and
Winkler showed a polynomial time bound on the mixing time when $k=3$
and Miracle and Streib~\cite{MS} generalized this to all constant
$k$. They proved a
bound of $\Omega(n^{-2(k-1)})$ on the spectral gap 
of $\mk$, which after applying the comparison
technique~\cite{RT98} and relating the gap to the mixing time, leads to a
bound of $O(n^{2k+6} \ln k)$ on the mixing time of the 
permutation process. 

Here we assume there are at least $C_q\log n$ particles of each 
type (where $C_q$ is a constant depending on the minimum bias $q$).
We show that at each application of the iterated decomposition given 
in~\cite{MS}, the chains are $1/n$-orthogonal.
Thus, we can apply Theorem~\ref{thm:indep} iteratively to get a bound 
of $\Omega(n^{-2})$ on the spectral gap of the $k$-particle process 
$\mk$.  Thus,  for $k$ as large as $\Theta(n/\log n)$, we obtain nearly optimal bounds on the spectral gap of $\mk$ and the first polynomial time bound on
the mixing time of $\mn$. This is a dramatic 
improvement over the state of the art.

\section{Preliminaries}\label{sec:prelim} 
We first fix some notation and terminology.  For a positive integer 
$n$, we write $[n]$ to mean $\{1,2,\dots,n\}$.  We also write $I_n$ to 
mean the $n \times n$ identity matrix.  The ``top eigenvalue'' of a matrix
is the largest eigenvalue in absolute value, and the ``top eigenvector''
is its corresponding eigenvector.  
The symbol $\otimes$ is used
for tensor product.  We write $(v)_i$ to mean the $i^{\text{th}}$ coordinate
of a vector $v$.  For any Markov chain $\m$, the ``gap'', denoted $\Gap(\m)$,
is the difference of 1 and the second largest eigenvalue of the transition
matrix of $\m$.

Our argument begins with the same setup as in the Disjoint 
Decomposition theorem of~\cite{mr}.  We assume $\m$ is an ergodic
Markov chain over a finite state space $\Omega$ with transition matrix $P$ that is reversible with
respect to the stationary distribution $\pi$; that is, it satisfies
the following \emph{detailed balance} condition: for all $\sigma,\tau\in\Omega$, 
$\pi(\sigma)P(\sigma,\tau)=\pi(\tau)P(\tau,\sigma).$
Let $\Omega = \cup_{i=1}^{\rhat} 
\Omegahat_i$ be a partition of the state space into $\rhat$ disjoint 
pieces.  For each $i \in [\rhat]$, define $\Phat_{i} = 
\transmat(\Omegahat_i)$ as the restriction of $\transmat$ to 
$\Omegahat_i$ which rejects moves that leave $\Omegahat_i$. In 
particular, the restriction to $\Omegahat_i$ is a Markov chain, 
$\hat{\m}_{i}$, with state space $\Omegahat_i$, where the transition 
matrix $\Phat_{i}$ is defined as follows: if $\elemone \not= \elemtwo$ 
and $\elemone,\elemtwo \in \Omegahat_i$ then 
$\Phat_i(\elemone,\elemtwo)=\transmat(\elemone,\elemtwo)$; if $\elemone 
\in \Omegahat_i$ then $\Phat_i(\elemone,\elemone)=1-\sum_{\elemtwo \in 
\Omegahat_i, \elemtwo \not= \elemone} \Phat_i(\elemone,\elemtwo)$. We 
call the $\hat{\m}_{i}$ \emph{restriction} chains.  Let $\pihat_i$ be 
the normalized restriction of $\pi$ to $\Omegahat_i$; i.e. 
$\pihat_i(S)=\pi(S\cap\Omegahat_i)/\pi(\Omegahat_i)$ for any $S 
\subseteq \Omega$.  The second largest eigenvalue of $\Phat_i$ will be
denoted $\lambda_i$, and $\lambda_{\max} = \max_i \lambda_i$.

Define $\Pbar$ to be the aggregated 
transition matrix on the state space $[\rhat]$ defined by 
$\Pbar(i,j)=\pi(\Omegahat_i)^{-1}\sum_{{\elemone\in\Omegahat_i,}{\elemtwo
\in \Omegahat_j}} \pi(\elemone)\transmat(\elemone,\elemtwo)$.  Then 
$\Pbar$ is the transition matrix of a reversible Markov chain $\bar{\m}$
with respect to the measure $\pibar$ defined by $\pibar(i) := 
\pi(\Omegahat_i)$.  We call $\bar{\m}$ the \emph{projection} chain.



It is useful to decompose the matrix $\transmat$ into the part that 
performs restriction moves and the part that performs all other moves. 
Define $\Phat$ as the block diagonal $|\Omega|\times |\Omega|$ matrix 
with the $\Phat_i$ matrices along the diagonal; i.e. $\Phat$ is 
obtained from $\transmat$ by rejecting moves between different parts of 
the partition.  Define $\Ptilde$ to be the transition matrix of the 
Markov chain defined by rejecting moves from $\elemone$ to $\elemtwo$ 
if $\elemone$ and $\elemtwo$ are within the same $\Omegahat_i$.  
Then $(\Phat+\Ptilde)(\elemone,\elemtwo)=\transmat(\elemone,\elemtwo)$ 
unless $\elemone=\elemtwo$, and 
$(\Phat+\Ptilde)(\elemone,\elemone)=\transmat(\elemone,\elemone)+1,$ 
since each move of $\transmat$ gets rejected in exactly one of $\Phat$ 
or $\Ptilde$ (and of course the probability of transitioning from a 
state is 1).  Therefore, we have $\transmat = 
\Phat+\Ptilde-I_{|\Omega|}$.

Note that for any pair $\elemone,\elemtwo \in \Omega$, the transitions 
$(\elemone,\elemtwo)$ and $(\elemtwo,\elemone)$ are either both nonzero 
in $\Ptilde$ or both zero in $\Ptilde$. Thus $\Ptilde$ is itself the 
disjoint union of a set of ergodic, reversible Markov chains $\Ptilde_1, \Ptilde_2, 
\dots, \Ptilde_{\rtilde}$ on state spaces $\Omegatilde_1, 
\Omegatilde_2, \dots, \Omegatilde_{\rtilde}$.  We call these chains 
\emph{complementary restrictions}.


In order to prove our decomposition results, we wish to apply 
Equation~\ref{eq:eigen} to $\transmat$.  However, since $\transmat$ may 
not be symmetric, we appeal to the following symmetrization technique 
that appears in~\cite[p. 153]{LPW06}. 
Given $\transmat$ with stationary distribution $\pi$, 
define a matrix $\A := \A(\transmat)$ by $\A(\elemone,\elemtwo) 
:=\pi(\elemone)^{1/2}\pi(\elemtwo)^{-1/2}\transmat(\elemone,\elemtwo)$. 
$\A$ is similar to $\transmat$ (i.e. they have the same 
eigenvalues), but is symmetric, so we can infer a bound on the second 
eigenvalue of $\transmat$ by applying Equation~\ref{eq:eigen} to $\A$. It 
is easy to check that the top eigenvector of $\A$ 
is $\sqrt{\pi}$, which is the vector with entries 
$\sqrt{\pi(\elemone)}$ for any $\elemone\in\Omega$. 


We apply this same symmetrization technique to other matrices as well. 
For $i \in [\rhat]$ we let $\Ahat_i := \A(\Phat_i)$ and for $i \in 
[\rtilde]$ we let $\Atilde_i := \A(\Ptilde_i)$.  We then write $\Ahat$ 
to mean the $|\Omega| \times |\Omega|$ matrix with 
$\Ahat(\elemone,\elemtwo) = \Ahat_i(\elemone,\elemtwo)$ if 
$\elemone,\elemtwo \in \Omegahat_i$ for some $i \in [\rhat]$, and zero 
otherwise. Analogously, we write $\Atilde$ to mean the $|\Omega| \times 
|\Omega|$ matrix with $\Atilde(\elemone,\elemtwo) = 
\Atilde_i(\elemone,\elemtwo)$ if $\elemone,\elemtwo \in \Omegatilde_i$ 
for some $i \in [\rtilde]$, and zero otherwise. It is important to note 
that $\Ahat \not= \A(\Phat)$ and 
$\Atilde \not= \A(\Ptilde)$.  This allows us to write $\A = \Ahat+\Atilde-I_{|\Omega|}.$  See Proposition~\ref{prop:matdecomp} in Section~\ref{sec:prop:matdecomp}.

\section{First Decomposition Bound}\label{sec:decomp}

We wish to apply Equation~\ref{eq:eigen} to $\A$.  Recall that 
$\sqrt{\pi}$ is the top eigenvector of $\A$.  Let $x \in \R^{|\Omega|}$ 
with $x\perp \sqrt{\pi}$ and $\|x\|=1$. We will decompose $x$ into two vectors $\Par{\xhat}$ and 
$\Perp{\xhat}$ as follows (note: this is similar to the vector 
decomposition used for the Zig Zag Product in~\cite{RVW}). For any 
$i\in [\rhat]$, let $\xhat_i\in\R^{|\Omegahat_i|}$ be the vector 
defined by $\xhat_i(\elemone)=\xhat(\elemone)$ for all $\elemone \in 
\Omegahat_i$. Then $x=\sum_i e_i\otimes \xhat_i$.  We further decompose 
$\xhat_i$ into $\Par{\xhat_i}$, the part that is parallel to 
$\sqrt{\pihat_i}$, and $\Perp{\xhat_i}$, the part that is perpendicular 
to $\sqrt{\pihat_i}$; recall that $\sqrt{\pihat_i}$ is the top 
eigenvector of $\Ahat$. Finally, define $\Par{\xhat}, \Perp{\xhat} \in 
\R^{|\Omega|}$ by $\Par{\xhat}=\sum_i e_i\otimes \Par{\xhat_i}$ and 
$\Perp{\xhat}=\sum_i e_i\otimes\Perp{\xhat_i}$.  Hence $x = \sum_i 
e_i\otimes \xhat_i = \Par{\xhat}+\Perp{\xhat}$. Define $\Par{\xtilde}$ and $\Perp{\xtilde}$ analogously.  

The next lemma makes concrete the intuition that if a 
particular distribution is far from stationarity, then it will either 
be far from stationarity on some restriction---in which case $\Ahat$ 
will bring it closer to stationarity (as in part 1)---or on the projection---in which 
case $\Atilde$ will bring it closer to stationarity (as in part 2).
The key quantities needed to prove Theorem~\ref{thm:main} are $\deltahpeh, \deltahpet,$ and $\deltahpat$:
\[ \deltahpeh = \frac{\langle\Perp{\xhat},\Perp{\xhat}\Ahat\rangle}{\|\Perp{\xhat}\|^2},~ 
\deltahpet = \frac{\langle\Perp{\xhat},\Perp{\xhat}\Atilde\rangle}{\|\Perp{\xhat}\|^2},~
\deltahpah = \frac{\langle\Par{\xhat},\Par{\xhat}\Ahat\rangle}{\|\Par{\xhat}\|^2},~
\deltahpat = \frac{\langle\Par{\xhat},\Par{\xhat}\Atilde\rangle}{\|\Par{\xhat}\|^2}. \] 
\begin{Lemma}\label{lem:eta}
With the above notation, (1.)  $\deltahpeh \leq \lambda_{\max}$, and (2.) $\deltahpat \leq \lambdabar.$
\end{Lemma}
\noindent  The proof is in Section~\ref{sec:deferredMain}.  We sketch
here the main idea for part (2).  Recall $\Par{\xhat_i}$ is parallel
to $\sqrt{\pihat_i}$ for all $i \in [\rhat]$. Define $\alpha_i \in \R$
by $\Par{\xhat_i} = \alpha_i\sqrt{\pihat_i}$, and let $\alpha \in
\R^{\rhat}$ be the vector with $(\alpha)_i = \alpha_i$.
We want to show that the effect of $\Atilde$ on $\Par{\xhat}$ 
is similar to the effect of $\Abar$ on $\alpha$.  Specifically, we
show $\|\Par{\xhat}\|^2=\|\alpha\|^2$ and $\langle
\Par{\xhat},\Par{\xhat}\Atilde\rangle =
\langle\alpha,\alpha\Abar\rangle$.
Applying Equation~\ref{eq:eigen} to $\Abar$, this shows
\vspace{-.1in}
\[\deltahpat =\frac{\langle \Par{\xhat},\Par{\xhat}\Atilde\rangle}{\|\Par{\xhat}\|^2} = \frac{\langle\alpha,\alpha\Abar\rangle}{\|\alpha\|^2} \leq \bar{\lambda}.\]
\vspace{-.1in}

If $\Par{\xhat}\A$ were orthogonal to $\Perp{\xhat}$, then Lemma~\ref{lem:eta} would be sufficient.  However, 
this is not true in general and
so it is necessary to bound the cross terms generated by applying the
matrix $\Atilde$ to $\Par{\xhat}$.

\begin{Lemma}\label{lem:mixed}
With the above notation, 
$|\langle\Perp{\xhat},\Par{\xhat}\Atilde\rangle| \leq 
\sqrt{(1-\deltahpat)(1-\deltahpet)}\|\Par{\xhat}\|\|\Perp{\xhat}\|.$
\end{Lemma}
\noindent In order to prove Lemma~\ref{lem:mixed}, we
let $\{\lambdatilde_i\}$ be the eigenvalues of $\Atilde$ with 
corresponding eigenvectors $\{\vtilde_i\}$.  As $\Atilde$ is symmetric, the real spectral 
theorem tells us that its eigenvectors form an orthonormal basis of 
$\R^{|\Omega|}$.  We consider the basis representations $\Perp{\xhat} = \sum_i\Perp{a_i}\vtilde_i$ and $\Par{\xhat} = 
\sum_i\Par{a_i}\vtilde_i$ and define vectors $\Par{\bigx}:=\sum_i \sqrt{1-\lambdatilde_i}\Par{a_i}\vtilde_i$ and
$\Perp{\bigx}:= \sum_i \sqrt{1-\lambdatilde_i}\Perp{a_i}\vtilde_i$ which satisfy $\|\Par{\bigx}\|^2= (1-\deltahpat)\|\Par{\xhat}\|^2$ and  $\|\Perp{\bigx}\|^2 = (1-\deltahpet)\|\Perp{\xhat}\|^2$. Then we show $\langle \Perp{\xhat},\Par{\xhat}\A \rangle= -\langle \Perp{\bigx},\Par{\bigx}\rangle$, and we finish by applying Cauchy-Schwartz.  The details are in Section~\ref{sec:deferredMain}.
Finally, we are ready to prove Theorem~\ref{thm:main}.  

\vspace{.1in}

\noindent\emph{Proof of Theorem~\ref{thm:main}}.

Let $x\in \R^{|\Omega|}$ with $x\perp \sqrt{\pi}$ and $\|x\|=1$.  By Equation~\ref{eq:eigen}, 
$\Gap(\m) \geq 1-\langle x,x\A\rangle$.  
Applying Lemma~\ref{lem:xparA} and the definitions of $\deltahpeh$, $\deltahpat$,
and $\deltahpet$, Equation~\ref{eq:split} becomes
\begin{align}
 \langle {x},{x}A\rangle&=  \langle \Par{\xhat},\Par{\xhat}\Atilde\rangle +2
  \langle\Perp{\xhat},\Par{\xhat}\Atilde\rangle
  +\langle\Perp{\xhat},\Perp{\xhat}(\Ahat+\Atilde-I_{|\Omega|})\rangle\nonumber\\
  &= \deltahpat \|\Par{\xhat}\|^2 +2 \langle\Perp{\xhat},\Par{\xhat}\Atilde\rangle
  +(\deltahpeh +\deltahpet -1)\|\Perp{\xhat}\|^2.\label{eq:strongestbound}
\end{align}
Applying Lemma~\ref{lem:mixed}, we have
\[ \Gap(\m) \geq 1 - (\deltahpat \|\Par{\xhat}\|^2 + 
2 \sqrt{(1-\deltahpat)(1-\deltahpet)}\|\Par{\xhat}\|\|\Perp{\xhat}\|
+(\deltahpeh +\deltahpet -1)\|\Perp{\xhat}\|^2).
\]
Rearranging terms and using the fact that
$1=\|x\|^2=\|\Perp{x}\|^2+\|\Par{x}\|^2$, we have
\begin{equation}\label{eq:main}
\Gap(\m) \geq \min_{x\perp\sqrt{\pi}:\|x\|=1} 
(1-\deltahpeh)\|\Perp{\xhat}\|^2+\left(\sqrt{1-\deltahpet}\|\Perp{\xhat}\| - \sqrt{1-\deltahpat}\|\Par{\xhat}\|\right)^2.
\end{equation}
Finally, we use the fact that Equation~\ref{eq:main} is minimized when $\deltahpeh$ and
$\deltahpat$ are maximized, together with Lemma~\ref{lem:eta} to prove the following. See Section~\ref{sec:minimized} for the remaining details.
\begin{equation}\label{eq:main2}
\Gap(\m) \geq \min_{p^2+q^2=1}\gmin q^2+\left(q\sqrt{1-\deltahpet} - p\sqrt{\gammabar}\right)^2.
\end{equation}
\vspace{-.1in}
\qed\\

The statement of Theorem~\ref{thm:main} is admittedly technical.
However, from it we can prove $\gamma\geq \gmin\gammabar/3$
(see Corollary~\ref{cor:compareproduct}) and Corollary~\ref{cor:compareT} in Section~\ref{sec:compare}, which is
identical to Theorem~1 from~\cite{jstv04} but applied to the spectral
gap instead of the Poincar\'{e} constant; it states  $\gamma\geq \min\left\{\frac{\gammabar}{3}, \frac{\gmin\gammabar}{3T+\gammabar}\right\}$.
Both results use a trivial bound on $\deltahpet$ of the minimum
eigenvalue of $\Atilde$. 

\vspace{-.1in}

\section{Complementary Decomposition}\label{sec:tight}\label{sec:orthog}
We use the technology developed in 
Section~\ref{sec:decomp} to prove Theorem~\ref{thm:primaldual} in Section~\ref{sec:compDec}.
To show how to apply Theorem~\ref{thm:primaldual}, we develop the notion of $\epsilon$-orthogonality 
that captures the key relationship between the top eigenvectors of $\Ahat$ 
and $\Atilde$ that allows us to get good bounds on $\Gap(\m)$.
By Theorem~\ref{thm:primaldual}, if $\gminhat$ and $\gmintilde$ are not 
too small, it suffices to show that $\|\Perp{\xhat}\|^2$ and 
$\|\Perp{\xtilde}\|^2$ cannot both be small.  To this end, we further 
decompose $\Perp{\xhat}$ and $\Par{\xtilde}$ based on the eigenvectors 
of $\Atilde$.  Define $S=\{i:\mu_i=1\}$ and vectors $\b=\sum_{i\in S} \Par{a_i}v_i$ and 
$\a=\sum_{i\notin S} \Par{a_i}v_i$. Similarly, let $\d=\sum_{i\in S} 
\Perp{a_i}v_i$ and $\c=\sum_{i\notin S} \Perp{a_i}v_i$. Notice 
$\Par{\xtilde} = \b+\d$ and $\Perp{\xtilde} =\a+\c$, so that the 
vectors in each row (resp. column) of the following table sum to the 
vector in its row (resp. column) label.
\begin{center}
\begin{tabular}{c|c|c|}
\multicolumn{1}{c}{}
 & \multicolumn{1}{c}{$\Par{\xtilde}$}
 & \multicolumn{1}{c}{$\Perp{\xtilde}$} \\
\cline{2-3}
$\Par{\xhat}$ &  $\b$ & $\a$\\
\cline{2-3}
$\Perp{\xhat}$ & $\d$ & $\c$\\
\cline{2-3}
\end{tabular}
\end{center}
The vectors within each row are orthogonal, as they are in the span of 
eigenvectors with distinct eigenvalues.  However, the vectors within 
each column are not necessarily orthogonal. 

We say that $\Phat_1, \Phat_2, \ldots,\Phat_{\rhat}$ and $\Ptilde_1, 
\Ptilde_2, \ldots, \Ptilde_{\rtilde}$ is an \emph{$\epsilon$-orthogonal} 
decomposition of $\m$ if $\|\b\|^2\leq \epsilon^2$.  As shorthand, we
may also say that the decomposition is $\epsilon$-orthogonal if
the same condition holds.  The idea is that $\epsilon$-orthogonality implies that if a distribution is far from stationarity then it will either be far from stationarity on some restriction or on some complementary restriction.

To prove Theorem~\ref{thm:indep} from Theorem~\ref{thm:primaldual}, 
 we must show that if $\|\b\|^2\leq \epsilon^2$, then $\|\Perp{\xhat}\|^2+ \|\Perp{\xtilde}\|^2\geq
(1-\epsilon)^2$.  As the sum of the squared norms of the vectors in the above table is 1, it is reasonable to expect that if $\|\b\|^2$ is small, then $\|\Perp{\xhat}\|^2+ \|\Perp{\xtilde}\|^2$ is large.  However, this is not as straightforward as one might expect, as the vectors within each column are not necessarily orthogonal, so we may have $\|\Perp{\xtilde}\|^2< \|\a\|^2 + \|\c\|^2$.  The full proof appears in Section~\ref{sec:compDec}.


\subsection{Bounding $\epsilon$}\label{sec:indep}
In order to apply Theorem~\ref{thm:indep}, one must determine an $\epsilon$ for which $\Ahat$ and $\Atilde$ are $\epsilon$-orthogonal.  In other words, one must
bound $\|\b\|^2$.  In this section, we provide one method for doing so.
  The proof of the following lemma appears in Section~\ref{sec:r}.

\begin{Lemma}\label{lem:directproduct}
 Define $r(i,j)=\frac{\pi(\Omegahat_i\cap\Omegatilde_j)}{\pi(\Omegahat_i)\pi(\Omegatilde_j)}.$
Then 
  $\displaystyle \|\b\|^2\leq \sum_{(i,j)}\pi(\Omegahat_i\cap\Omegatilde_j)(\sqrt{r(i,j)}-1/\sqrt{r(i,j)})^2.$
\end{Lemma}

In our application in Section~\ref{sec:perm}, the Markov chain $\m$
acts on a product space, where $\Omega = \Omega_1\times\Omega_2$.
In other words, for every $i\in[\rhat]$ and $j\in[\rtilde]$, there is
a unique $\sigma\in\Omegahat_i\cap\Omegatilde_j$.  Abusing 
notation, we write $\sigma = (i,j)$.
In this sense, it is similar to the direct product of independent
Markov chains, but the transition probabilities are
not necessarily independent.  The function $r(i,j)$ allows us to
capture the degree of dependence.  
In this setting, we rewrite $r(i,j)$ as
 $r(i,j)=\frac{\pi(i,j)}{\pi(\Omegahat_i)\pi(\Omegatilde_j)}.$
Then 
  \begin{equation}\label{eq:productchains}
  \|\b\|^2\leq \sum_{(i,j)}\pi(i,j)(\sqrt{r(i,j)}-1/\sqrt{r(i,j)})^2.
  \end{equation}

\noindent Notice that $r(i,j)=1$ for all $i,j$
when $\m$ is the direct product of 2 independent Markov chains.  Thus, by iterating, we can
immediately prove the following (well-known) result.

  \begin{Corollary}\label{cor:product}
    If $\m$ is the direct product of $N$ Markov chains $\{\m_i\},$
    then  $\Gap(\m)=\min_i\Gap(\m_i).$
\end{Corollary}

\newcommand{\area}{N}
\newcommand{\hiA}{a^*}
\section{Application to Permutations}\label{sec:perm}
Next, we will analyze the nearest neighbor Markov chain $\mn$ over biased permutations.  Theorem~\ref{thm:indep} and Lemma~\ref{lem:directproduct} are key to our improved bounds.
Let $\Omega = S_n$ be the set of all permutations $\sigma = (\sigma(1), \ldots, \sigma(n))$ and $\cP$ be a set of probabilities $p_{i,j} \in [0,1]$ for $1\leq i \neq j \leq n$ with $p_{j,i} = 1-p_{i,j}.$  Here we consider the case where the probabilities $\cP$ form a weakly monotone bounded $k$-class (introduced by~\cite{MS}, see Section~\ref{sec:MKprelim} for more details).  In a bounded $k$-class, the set $[n]$ can be partitioned into $k$ particle classes $C_1, \ldots C_k,$ elements from the same class are swapped with probability $1/2,$ and each element from class $C_i$ and each element from $C_j$ (with $i<j$) are put in increasing order with the same probability $p_{i,j}.$  Additionally, $p_{i,j}/p_{j,i} \geq 1/q$ for all $i<j$ for some constant $q <1.$  As in~\cite{MS, HW16} we first analyze an auxiliary chain $\mk$ on $k$-particle processes where elements in the same class are indistinguishable and then use comparison techniques~\cite{dsc,RT98} to analyze $\mn.$ The chain $\mk$ exchanges particles in different classes across elements in smaller particle classes.  We use the same techniques as~\cite{MS}, but we get much better results by applying the Complementary Decomposition Theorem, Theorem~\ref{thm:indep}.  The details are given in Section~\ref{sec:mn}.  Let $q_{i,j} = p_{i,j}/p_{j,i}.$  Below we write $C(i):=C_i$ for notational convenience.

\vspace{.1in}
\noindent  {\bf The particle process Markov chain $\mk$} 

\vspace{.05in}
\noindent {\tt Starting at any permutation $\sigma_0$, iterate the following:} 
\begin{itemize}
\item At time $t,$ choose a position $1\leq i \leq n$ and direction $d \in \{L, R\}$ uniformly at random.  
\item If $d = L$, find the largest $j$ with $1 \leq j<i$ and $\C(\sigma_t(j)) \geq \C(\sigma_t(i))$ (if one exists).  If $\C(\sigma_t(j)) > \C(\sigma_t(i)),$ then 
with probability $1/2$ exchange $\sigma_t(i)$ and $\sigma_t(j)$ to obtain $\sigma_{t+1}.$  
\item If $d = R$, find the smallest $j$ with $n\geq j>i$ and $\C(\sigma_t(j)) \geq \C(\sigma_t(i))$ (if one exists).  If $\C(\sigma_t(j)) > \C(\sigma_t(i)),$ then exchange $\sigma_t(i)$ and $\sigma_t(j)$ to obtain $\sigma_{t+1}$  with probability $$\frac{1}{2}~q_{\sigma_t(j),\sigma_t(i)}  \prod_{i<k<j}\left(q_{\sigma_t(j),\sigma_t(k)}q_{\sigma_t(k),\sigma_t(i)}\right).$$
 
\item Otherwise, do nothing so that $\sigma_{t+1} = \sigma_t.$
\end{itemize}

\noindent 
Letting $\psize = \max\left\{\frac{\log (6n^2)+\log(1+2/(q^{-1}-1))}{\log(1+q^{-1})-1}, \frac{\log^2(c_1)}{\log(2/(q(1+q^{-1})))}\right\} = \Theta(\log n),$ we  prove the following.

\begin{theorem}\label{pprocess}
If the probabilities $\transmat$ are weakly monotonic, form a bounded $k$-class for $k\geq 2,$ and  $|C_i|\geq 2\psize$ for $1\leq i \leq k$ then the spectral gap $\text{Gap}(\transmat)$ of the chain $\mk$ satisfies
\abovedisplayskip=3pt
\belowdisplayskip=3pt
$\text{Gap}(\transmat) =\Omega(n^{-2}).$
\end{theorem}

The proof of Theorem~\ref{pprocess} uses the same inductive technique as~\cite{MS} where at each level of the induction we will fix the locations of particles in one less particle class.  
For $i\geq 0$, let $\sigma_i$ represent an arbitrary fixed location of
the particles in classes $C_1,C_2,\ldots, C_i$ (notice that $\sigma_0=\sigma$
represents no restriction).  
For example, let $\sigma_{2} = 12\_1\_2\_1\_\_,$ where the $\_$ represents an empty location that can be filled with particles in class $C_3$ or higher.  
We will consider a smaller chain $\m_{\sigma_{i}}$ whose state space $\Omega_{\sigma_i}$
is the set of all configurations where the elements in classes
$C_1,\ldots, C_i$ are in the locations given by $\sigma_i$.
  The moves of $\m_{\sigma_{i}}$
are a subset of the moves of $\mk$.  It rejects all moves of $\mk$
involving an element of $C_1,C_2,\ldots, C_i$.
We prove by induction that  $\m_{\sigma_{i}}$ has spectral gap
$\Omega(n^{-2}(1-1/n)^{2(k-2-i)})$
  for all choices of
 $\sigma_i$ (given a fixed $i$).  
 Specifically, we will assume that $\m_{\sigma_{i}}$ is rapidly mixing
 by induction, and we will use that fact to 
 prove that $\m_{\sigma_{i-1}}$ is rapidly mixing.

 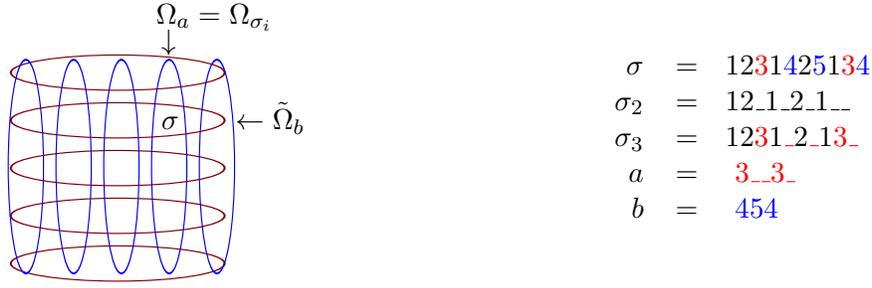
\begin{figure}
   \centering
   \begin{minipage}[c]{.5\textwidth}
     \centering
   \includegraphics[scale=.18]{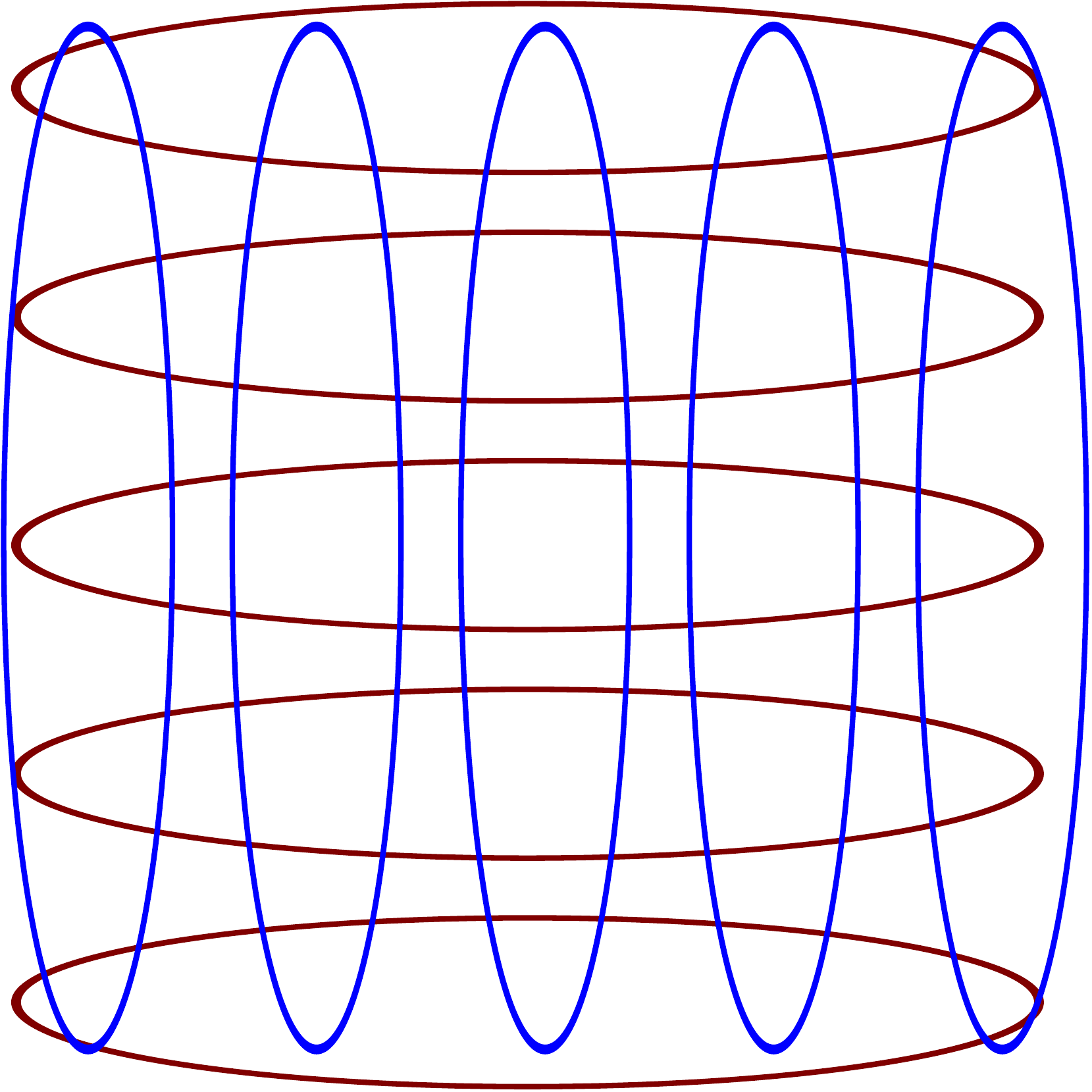}
   \put(-30,98){$\Omegahat_{a}=\Omega_{\sigma_i}$}
   \put(-28,88){$\downarrow$}
   \put(-28,58){$\sigma $}
  \put(0,58){$\leftarrow\Omegatilde_b$}
   \end{minipage}\hfill
   \begin{minipage}[c]{.5\textwidth}
     \centering
     \begin{tabular}{rl}
     $\sigma$ &$=~~
     12{\color{red}3}1{\color{blue}4}2{\color{blue}5}1{\color{red}3}{\color{blue}4}$\\
     $\sigma_2$ &$=~~ 12\_1\_2\_1\_\_$\\
     $\sigma_3$ &$=~~
     12{\color{red}3}1{\color{red}\_}2{\color{red}\_}1{\color{red}3\_}$\\
     $a$ &$= ~~~ {\color{red} 3\_\_3\_}$\\
     $b$ &$= ~~~ {\color{blue} 454}$
     \end{tabular}
     \end{minipage}
   \caption{The state space $\Omega_{\sigma_{i-1}}$ decomposed}
   \label{fig:product}
 \end{figure}

The state space $\Omega_{\sigma_{i-1}}$ of $\m_{\sigma_{i-1}}$ is a
product space.  Given any $\sigma$ consistent with
$\sigma_{i-1}$ (i.e. they agree on the locations of all elements in
classes $C_1,C_2,\ldots, C_{i-1}$), we write $\sigma=(a,b)$ as
follows.  Let $a$ be the 2-particle system obtained from $\sigma_i$ by removing particles of type less than $i$ (see Figure~\ref{fig:product}).   Let $A$ be the set of all possible choices of $a$; $A$ consists of all 2-particle systems with $|C_i|$ particles of type $i$ and $\sum_{j=i+1}^k|C_j|$ particles of type $\_$ (these are in bijection with staircase walks as in Figure~\ref{fig:staircase}).  Define $b$ to be the permutation of the elements of type bigger than $i$ that is consistent with $\sigma$; let $B$ be the set of all possible $b$.  Then, given any $(a,b)$ pair, there is
a unique $\sigma\in\Omega_{\sigma_{i-1}}$ corresponding to them.  This shows that
$\Omega_{\sigma_{i-1}}$ is a product space.  Moreover, the moves of
$\m_{\sigma_i}$ fix $a$ and perform $(j_1,j_2)$ transpositions, where
$j_1,j_2>i$; i.e. they operate exclusively on $b$. Thus, the
Markov chain $\m_{\sigma_i}$ is a restriction Markov chain of
$\m_{\sigma_{i-1}}$ with state space $\Omegahat_{a}$, which is the set
of $\sigma$ consistent with $a=\sigma_{i}$.
 On the other hand, the remaining moves of $\m_{\sigma_{i-1}}$ which
 form the complementary restrictions are $(i,j)$ transpositions, where
 $j>i$.  In other words, they fix $b$ and change $a$, so we label the state
 space of such a Markov chain $\Omegatilde_b$.  As these moves fix the
 relative order of all particles of type bigger than $i$, the
 complementary restriction chains can be seen as bounded exclusion
 processes on particles of type $i$ with  particles of type bigger
 than~$i$.  Bounded biased exclusion processes operate on staircase walks as in Figure~\ref{fig:staircase}, where every square has a different bias but they are all bounded by some $q$. We prove the following lemma in Section~\ref{sec:mk}.

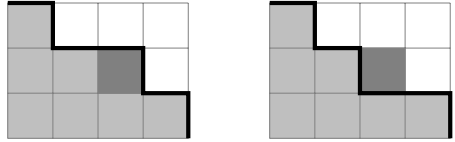
\begin{figure}[h]
  \centering
   \begin{tikzpicture}[scale=.6]
     \draw [fill=lightgray,thin,lightgray] (0,0) rectangle (1,3);
     \draw [fill=lightgray,thin,lightgray] (1,0) rectangle (2,2);
     \draw [fill=lightgray,thin,lightgray] (2,0) rectangle (4,1);
              \draw [fill=gray,thin,gray] (2,1) rectangle (3,2);
     \draw [help lines] (0,0) grid (4,3);
     \draw [ultra thick] (0,3) -- (1,3) -- (1,2) -- (3,2) -- (3,1) -- (4,1) -- (4,0);
   \end{tikzpicture}  
\hspace{2em}
    \begin{tikzpicture}[scale=.6]
     \draw [fill=lightgray,thin,lightgray] (0,0) rectangle (1,3);
     \draw [fill=lightgray,thin,lightgray] (1,0) rectangle (2,2);
     \draw [fill=lightgray,thin,lightgray] (2,0) rectangle (4,1);
         \draw [fill=gray,thin,gray] (2,1) rectangle (3,2);
     \draw [help lines] (0,0) grid (4,3);
     \draw [ultra thick] (0,3) -- (1,3) -- (1,2) -- (2,2) -- (2,1) -- (4,1) -- (4,0);
   \end{tikzpicture}
\caption{An exclusion process on staircase walks operates by adding or removing a square. } \label{fig:staircase}
\end{figure}

\begin{Lemma}\label{lem:dual}
The complementary restrictions at each level of the induction are bounded generalized biased exclusion processes with spectral gap $\Omega(n^{-2}).$
\end{Lemma}

The chain $\m_{\sigma_{i-1}}$ is not the direct product of
the chains on $a$ and $b$ because, e.g., $P((a,b),(a',b))$ depends on $b$.
However, we will show it is a $1/n$-orthogonal
decomposition by bounding $r(a,b)$ defined in
Lemma~\ref{lem:directproduct}.  We will define ``good" $a$'s to be
those with fewer than $N^*=C_q\log n$ inversions and ``good" $b$'s
to be those with fewer than $N^*$ inversions involving particles
of type $i+1$.  
As there are at least $2N^*$ particles of type $i+1$, if $a$ and $b$ are both good, then $(a,b)$ has no inversions between $i$ and $j$ for $j>i+1$.  For such pairs,
$r(a,b)$ is very close to 1.  For all other pairs, we show $r(a,b)\pi(\Omegahat_b\cap\Omegatilde_{a})$ is small.  By viewing $b$ as a staircase walk on particles of type $i+1$ with particles of any higher type, we will see that for either $a$ or $b$, the probability it is bad is smaller than the weighted sum of all biased exclusion processes with more than $N^*$ inversions (equivalently, area $N^*$ under the curve).   Lemma~\ref{lem:partitions} (see Section~\ref{sec:ortho}) bounds this by $1/(6n^2).$  

\begin{Lemma}\label{lem:perms}
 If the probabilities $\p$ are weakly monotonic and form a bounded \kP\  with $|C_i|\geq 2\psize$ for all $1\leq i \leq k$ then at each step of the induction we have a $1/n$-orthogonal decomposition.
\end{Lemma}
\begin{proof}
Assume that at this step in the induction, all particles in $C_1,C_2,\ldots, C_{i-1}$ are fixed in the same position in all permutations according to some $\sigma_{i-1}.$  
The stationary distribution of $\m_{\sigma_{i-1}}$ is 
\begin{equation}\label{eq:statdist}
\pi(\sigma) =  \prod_{i < j: \sigma(i) > \sigma(j)}q_{\sigma(i), \sigma(j)}Z^{-1},
\end{equation}
where $Z$ is the normalizing constant $\sum_{\sigma \in \Omega_{\sigma_{i-1}}} \pi(\sigma).$
Let $a^*$ and $b^*$  be the highest weight elements in $A$ and $B$, respectively. 
Using these definitions, the permutation $(a^*,b^*)$ has the particles in $C_1,C_2,\ldots, C_{i-1}$ fixed according to $\sigma_{i-1}$ and all other higher particles in sorted order.   In our example, $b^* = 445,$ $(a,b^*) = 1231423145,$ $a^* = 33\_\_\_,$ $(a^*,b) = 1231324154,$ and $(a^*,b^*) = 1231324145.$ 

Next, we will decompose the product over inversions in Equation~\ref{eq:statdist} into several quantities.
Notice that $\pi(a^*,b^*)$ is the product over inversions that are in every $\sigma\in\Omega$, normalized by $Z_{\sigma_{i-1}}$.  Define $w(a) = \pi(a,b^*)/\pi(a^*,b^*).$  This is the product over inversions that are in every $\sigma\in\Omegahat_a$ that are not in every $\sigma\in\Omega$. Similarly, $w(b) = \pi(a^*,b)/\pi(a^*,b^*)$ is the product over inversions in every $\sigma\in\Omegatilde_b$, and $w(a,b) = \pi(a,b) \pi(a^*,b^*)/ (\pi(a,b^*)\pi(a^*,b))$ is the product over inversions in $\sigma=(a,b)$ beyond those that are required by being in $\Omegahat_a$ and $\Omegatilde_b$.  From these definitions, it is clear that $\pi(a,b) = w(a)w(b) w(a,b)\pi(a^*,b^*).$  We will prove in Lemma~\ref{lem:good} that if both $a$ and $b$ are good then $w(a,b) = 1;$ i.e. that the weight of $(a,b)$ is determined entirely by being in $\Omegahat_a$ and $\Omegatilde_b$.

Next, define $\ztilde =  \sum_{a} w(a),$  $\zhat =\sum_{b}w(b),$  and let $\ztilde' = \sum_{a \text{ good}} w(a)$ be the sum $\ztilde$ restricted to only good configurations $a$ and $\ztilde'' =  \sum_{a~\text{bad}} w(a)$ be the sum $\ztilde$ restricted to only bad configurations $a.$  Similarly, let $\zhat' = \sum_{b \text{ good}} w(b)$ and $\zhat'' =  \sum_{b \text{ bad}} w(b).$
Define $\epsilon_1= 1/(6n^2).$  By Lemma~\ref{lem:partitions}, $\epsilon_1$ is a bound on the total weight of staircase walks with area more than $N^*$ under the curve.  We will use this to bound the weight of the bad $a$'s contributing to $\ztilde$ and the weight of the bad $b$'s contributing to $\zhat$.  Specifically, we will prove the following lemma (see proof in Section~\ref{sec:ortho}).
\begin{Lemma}\label{lem:zBound}\ \  
  \begin{enumerate}
  \item $\ztilde'' =  \sum_{a \text{ bad}} w(a) \leq \epsilon_1$ and $\zhat'' =  \sum_{b \text{ bad}} w(b) \leq \epsilon_1\zhat.$
 \item For all $a,b$ we have $\pi(a)\leq w(a)\pi(a^*,b^*)\zhat$ and $\pi(b)\leq w(b)\pi(a^*,b^*)\ztilde.$
 \item For \emph{good} $a$ and \emph{good} $b$ we have $\pi(a) \geq w(a) \pi(a^*,b^*)\zhat'$ and $\pi(b) \geq w(b) \pi(a^*,b^*)\ztilde'.$
 \item For all $b$ we have $\pi(b) \geq w(b)\pi(a^*,b^*)w(a^*,b).$
    \end{enumerate}
\end{Lemma}
\noindent We bound $r(a,b)$  when $a$ and $b$ are both good using Lemma~\ref{lem:zBound} part (3) and (1) to show
\[  r(a,b)=\frac{\pi(a,b)}{\pi(a)\pi(b)}  \leq \frac{1}{\pi(a^*,b^*)\zhat'\ztilde'}
\leq \left(\frac{1}{(1-\epsilon_1)^2}\right) \frac{1}{\zhat\ztilde \pi(a^*,b^*)} 
\leq \frac{1}{(1-\epsilon_1)^2}.\]
\noindent Using Lemma~\ref{lem:zBound} part (2) we will show that when $a$ and $b$ are both good,
\begin{align*}
  r(a,b)&=\frac{w(a)w(b)\pi(a^*,b^*)}{\pi(a)\pi(b)}
  \geq \frac{w(a)w(b)\pi(a^*,b^*)}{
     \left(w(a)\pi(a^*,b^*)\zhat\right)\left(w(b)\pi(a^*,b^*)\ztilde\right)} = \frac{1}{
     \zhat\ztilde\pi(a^*,b^*)} \geq(1-\epsilon_1)^2.
\end{align*}
This allows us to show
  \[\sum_{a~\text{good},\atop
    b~\text{good}}\pi(a,b)\left(\sqrt{r(a,b)}-\frac{1}{\sqrt{r(a,b)}}\right)^2\leq
  5\epsilon_1^2.\]

  Next, we assume either $a$ or $b$ is bad.  We
  will show that the weight of these configurations is so small that
  it overcomes the fact that $r(a,b)$ may not be close to 1.  We use the loose bound
 \[\pi(a,b)\left(\sqrt{r(a,b)}-\frac{1}{\sqrt{r(a,b)}}\right)^2\leq
 \pi(a)\pi(b) + \frac{\pi(a,b)^2}{\pi(a)\pi(b)}.\]
 
In order to upper bound $\frac{\pi(a,b)^2}{\pi(a)\pi(b)}$, we will use Lemma~\ref{lem:zBound} to show $\pi(a,b)/\pi(a)\leq w(b)/\zhat'$ for \emph{good} $a$ and $\pi(a,b)/\pi(b)\leq w(a)$ for all $b.$  Then using Lemma~\ref{lem:zBound} (part 1), we show $\sum_{a,b} \frac{\pi(a,b)^2}{\pi(a)\pi(b)}\leq \epsilon_1$ for \emph{bad} $a$ and all $b$  and $\sum_{a,b} \frac{\pi(a,b)^2}{\pi(a)\pi(b)}\leq \epsilon_1/(1-\epsilon_1)$ for \emph{good} $a$ and \emph{bad} $b$.
 Next, we will bound $\pi(a)\pi(b)$ for \emph{bad} $a$ and all $b$ using Lemma~\ref{lem:zBound} (part 1) and Equation~\ref{prodBound} to obtain
 \[\sum_{a~\text{bad}, b} \pi(a)\pi(b) \leq  \sum_{a~\text{bad}, b} w(a)w(b)\pi(a^*,b^*)^2\ztilde \zhat
 \leq \epsilon_1 (1-\epsilon_1)^{-3}.\]
Finally, we bound $\pi(a)\pi(b)$ for $b$ \emph{bad} and all $a$ using Lemma~\ref{lem:zBound} part (1) and Equation~\ref{prodBound} to obtain
\[\sum_{a, b~\text{bad}} \pi(a)\pi(b) \leq   \sum_{a, b~\text{bad}} w(a)w(b)\pi(a^*,b^*)^2\ztilde \zhat \leq\epsilon_1 (1-\epsilon_1)^{-4}.\]
Putting this all together, we have, for $\epsilon_1\leq .18,$
  \[\sum_{a, b}\pi(a,b)\left(\sqrt{r(a,b)}-\frac{1}{\sqrt{r(a,b)}}\right)^2\leq
  5\epsilon_1^2 + \epsilon_1 + \frac{\epsilon_1}{1-\epsilon_1} + \frac{\epsilon_1}{(1-\epsilon_1)^3}+  \frac{\epsilon_1}{(1-\epsilon_1)^4}\leq
  6\epsilon_1 = 1/n^2.\]
This will show $\Atilde$ and $\Ahat$ are $1/n$-orthogonal.
  The remaining details are in Section~\ref{sec:ortho}.\end{proof}
\begin{Remark}\label{rem:eta}
It is worth pointing out that this decomposition of $\m_{\sigma_{i-1}}$ does not satisfy the regularity conditions of~\cite{jstv04} needed to obtain a better bound.  For any $\upsilon\in \Omegahat_j$, define \[\pihat_j^{j'}(\upsilon)=\pi_j(\upsilon)\frac{\sum_{\upsilon'\in\Omegahat_{j'}}P(\upsilon,\upsilon')}{\Pbar(j,j')}.\]
We need to bound $\pihat_j^{j'}(\upsilon)/\pi_j(\upsilon)$ for any $j,j',$ and $\upsilon\in\Omegahat_j$.   For example, let $\sigma_{2} = 12\_11111\_2\_1\_\_$ and $\sigma_3 = 12311111\_231\_\_$.
Notice that the two permutations $\upsilon_1=12311111423156$ and $\upsilon_2=12311111523146$ are in the same restriction $\Omegahat_j$ (i.e. they are both consistent with $\sigma_3$).  They each have a single move to $\Omegahat_{j'}$: the move of swapping the first 3 with the 4 (in the case of $\upsilon_1$) or 5 (in the case of $\upsilon_2$).  However, the probability of these moves differ by a factor of $(q_{4,3}/q_{5,3})(q_{4,1}/q_{5,1})^5,$ as there are five 1's between them.  In principle, there could be order $n$ smaller numbers between the two numbers we are swapping.  Thus, $\pihat_j^{j'}(\upsilon)/\pi_j(\upsilon)$ cannot be uniformly bounded to within $1\pm \eta$ unless $\eta$ is exponentially large.
\end{Remark}


\appendix

\section{Additional Details from the Decomposition Theorems}
In this section, we provide the full proofs that have been deferred from Section~\ref{sec:prelim}, Section~\ref{sec:decomp}, and Section~\ref{sec:tight}.

\subsection{Proof of Proposition~\ref{prop:matdecomp}}\label{sec:prop:matdecomp}

Next we will prove Proposition~\ref{prop:matdecomp}:
\begin{Proposition}\label{prop:matdecomp}
The matrix $\A$ satisfies 
$\A = \Ahat+\Atilde-I_{|\Omega|}.$
\end{Proposition}
\begin{proof}
Since the symmetrization operation is performed component-wise and 
doesn't change the diagonals of a matrix, it suffices to show that for 
any $\elemone \neq \elemtwo$, 
$\A(\elemone,\elemtwo)=\Ahat(\elemone,\elemtwo)+\Atilde(\elemone,\elemtwo)$.  
If $\elemone$ and $\elemtwo$ are both in $\Omegahat_i$ for some $i \in 
[\rhat]$, then $\Atilde(\elemone,\elemtwo)=0$ and 
\[\A(\elemone,\elemtwo) = 
\frac{\sqrt{\pi(\elemone)}}{\sqrt{\pi(\elemtwo)}} 
\transmat(\elemone,\elemtwo) = 
\frac{\sqrt{\pihat_i(\elemone)}}{\sqrt{\pihat_i(\elemtwo)}} 
\Phat_i(\elemone,\elemtwo) = \Ahat(\elemone,\elemtwo).\] Otherwise, 
$\elemone,\elemtwo \in\Omegatilde_i$ for some $i \in [\rtilde]$, and 
thus $\Ahat(\elemone,\elemtwo)=0$ and, analogously, 
$\A(\elemone,\elemtwo)=\Atilde(\elemone,\elemtwo)$. 
\end{proof}

\subsection{Proof of Lemmas~\ref{lem:eta} and~\ref{lem:mixed}}\label{sec:deferredMain}
In this section we will prove Lemmas~\ref{lem:eta} and~\ref{lem:mixed}.  First we prove $\Par{\xhat}\A=\Par{\xhat}\Atilde,$  which will be useful in the following proof.
\begin{Lemma}\label{lem:xparA}  
The following holds: $\Par{\xhat}\A=\Par{\xhat}\Atilde.$
\end{Lemma}
\begin{proof}
As $\A = \Ahat+\Atilde-I_{|\Omega|}$, we have  
$\Par{\xhat}\A = 
\Par{\xhat}\Ahat + \Par{\xhat}\Atilde - \Par{\xhat}I_{|\Omega|}.$
The statement follows if $\Par{\xhat}\Ahat=\Par{\xhat}$, as the first 
and third terms cancel.  To see this, note that 
$\Par{\xhat_i}\Ahat_i = \Par{\xhat_i}$ for all $i \in [\rhat]$, as 
$\Ahat_i$ has the same eigenvalues as $\Phat_i$ and $\Par{\xhat}$ is 
in the space spanned by eigenvectors of $\Phat_i$ with eigenvalue 1.
Thus
\vspace{-.1in}

\[\Par{\xhat}\Ahat = \left(\sum_i e_i \otimes \Par{\xhat_i}\right)\Ahat
                   = \sum_i \left(e_i \otimes \Par{\xhat_i}\right)\Ahat
                   = \sum_i e_i \otimes \Par{\xhat_i}\Ahat_i
                   = \sum_i e_i \otimes \Par{\xhat_i}
                   = \Par{\xhat}.\]
\vspace{-.1in}
\end{proof}

Next, we prove Lemma~\ref{lem:eta}(1), which states that $\deltahpeh \leq \lambda_{\max}.$  \\

\noindent\emph{Proof of Lemma~\ref{lem:eta}}(1).

We wish to bound 
$\frac{\langle\Perp{\xhat},\Perp{\xhat}\Ahat\rangle}{\|\Perp{\xhat}\|^2}$. 
By construction, $\Perp{\xhat_i}$ is perpendicular to $\sqrt{\pihat_i}$ 
for all $i \in [\rhat]$, which is the top eigenvector of $\Ahat_i$. By 
Equation~\ref{eq:eigen}, we know 
$\langle\Perp{\xhat_i},\Perp{\xhat_i}\Ahat_i\rangle \leq \lambda_i 
\|\Perp{\xhat_i}\|^2$. Thus, 
\[ \deltahpeh \|\Perp{\xhat}\|^2 
	= \langle\Perp{\xhat},\Perp{\xhat}\Ahat\rangle 
	= \sum_i\langle\Perp{\xhat_i},\Perp{\xhat_i}\Ahat_i\rangle 
	\leq \sum_i\lambda_i\|\Perp{\xhat_i}\|^2 
	~\leq \lambda_{\max}\|\Perp{\xhat}\|^2.\]
\qed

Next, we prove Lemma~\ref{lem:eta}(2), which states that $\deltahpat \leq \lambdabar.$  \\

\noindent\emph{Proof of Lemma~\ref{lem:eta}}(2).
Recall $\Par{\xhat_i}$ is parallel to $\sqrt{\pihat_i}$ for all $i \in 
[\rhat]$. Define $\alpha_i \in \R$ by $\Par{\xhat_i} = 
\alpha_i\sqrt{\pihat_i}$, and let $\alpha \in \R^{\rhat}$ be the vector 
with $(\alpha)_i = \alpha_i$. 
Using the fact that $\pihat_i$ is a probability distribution, notice
\[\|\Par{\xhat}\|^2 = \sum_i \sum_{\elemone\in\Omegahat_i}(\alpha_i\sqrt{\pihat_i(\elemone)})^2 
	= \sum_i \alpha_i^2 \sum_{\elemone\in\Omegahat_i}\pihat_i(\elemone)
	= \sum_i \alpha_i^2 
	= \|\alpha\|^2.\]
Let $\Abar=\A(\Pbar)$ and $i,j \in [\rhat]$.  Then
\[ \Abar(i,j) = \frac{\sqrt{\pi(\Omegahat_i)}}{\sqrt{\pi(\Omegahat_j)}}\Pbar(i,j)
	= \frac{\sqrt{\pi(\Omegahat_i)}}{\sqrt{\pi(\Omegahat_j)}}\frac{1}{\pi(\Omegahat_i)}\sum_{\elemone\in\Omegahat_i\atop \elemtwo\in\Omegahat_j}\pi(\elemone)\transmat(\elemone,\elemtwo)
	= \frac{1}{\sqrt{\pi(\Omegahat_j)\pi(\Omegahat_i)}}\sum_{\elemone\in\Omegahat_i\atop \elemtwo\in\Omegahat_j}\pi(\elemone)\transmat(\elemone,\elemtwo) .\]
and for any $\elemtwo \in\Omega_j$, we have
\[ (\Par{\xhat}\A)_{\elemtwo} = \sum_i\sum_{\elemone\in\Omegahat_i}\alpha_i\sqrt{\pihat_i(\elemone)}\A(\elemone,\elemtwo)
	= \sum_{i}\sum_{\elemone\in\Omegahat_i}\alpha_i\sqrt{\pihat_i(\elemone)}\frac{\sqrt{\pi(\elemone)}}{\sqrt{\pi(\elemtwo)}}\transmat(\elemone,\elemtwo)
	= \sum_{i}\frac{\alpha_i}{\sqrt{\pi(\Omegahat_i)}\sqrt{\pi(\elemtwo)}}\sum_{\elemone\in\Omegahat_i}\pi(\elemone)\transmat(\elemone,\elemtwo). \]
Therefore,
\begin{align}
\langle \Par{\xhat},\Par{\xhat}\A\rangle &=
 \sum_j\sum_{\elemtwo\in\Omegahat_j} \alpha_j\sqrt{\pihat_j(\elemtwo)}\sum_{i}\frac{\alpha_i}{\sqrt{\pi(\Omegahat_i)}\sqrt{\pi(\elemtwo)}}\sum_{\elemone\in\Omega_i}\pi(\elemone)\transmat(\elemone,\elemtwo)\nonumber\\
  &= \sum_{i,j}\frac{\alpha_i\alpha_j}{\sqrt{\pi(\Omegahat_j)\pi(\Omegahat_i)}}\sum_{\elemtwo\in\Omegahat_j\atop
   \elemone\in\Omega_i}\pi(\elemone)\transmat(\elemone,\elemtwo)\nonumber\\
  &= \sum_{i,j}\alpha_i\alpha_j\Abar(i,j)\nonumber\\
  &= \langle\alpha,\alpha\Abar\rangle ~\leq ~
 \bar{\lambda}\|\alpha\|^2 ~=~ \bar{\lambda}\|\Par{\xhat}\|^2.\nonumber
\end{align}
Recalling that $\deltahpat = \langle \Par{\xhat},\Par{\xhat}\A\rangle / \|\Par{\xhat}\|^2$
and that we showed $\Par{\xhat}\A=\Par{\xhat}\Atilde$ in Lemma~\ref{lem:xparA}, we are done.
\qed\\

Recall that for any $v \in \R^{|\Omega|}$, we can write 
$v=\sum_{i}a_i\vtilde_i$ for some constants $a_1,a_2,\ldots, 
a_{|\Omega|} \in \R$, and we have $\|v\|^2 = \sum_i 
a_i^2\|\vtilde_i\|^2,$ and $v\Atilde = \sum_i 
a_i\lambdatilde_i\vtilde_i$. We can now prove Lemma~\ref{lem:mixed}. \\

\noindent\emph{Proof of Lemma~\ref{lem:mixed}}.
Let $\Perp{\xhat} = \sum_i\Perp{a_i}\vtilde_i$ and $\Par{\xhat} = 
\sum_i\Par{a_i}\vtilde_i$ be the basis representations of 
$\Perp{\xhat}$ and $\Par{\xhat}$.  Since $\Perp{\xhat}$ is 
perpendicular to $\Par{\xhat}$, we have 
$0=\langle\Perp{\xhat},\Par{\xhat}\rangle = 
\sum_{i}\Perp{a_i}\Par{a_i}\|\vtilde_i\|^2$. Notice
\[ \langle\Perp{\xhat},\Par{\xhat}\Atilde\rangle = \sum_{i}\lambdatilde_i\Perp{a_i}\Par{a_i}\|\vtilde_i\|^2
	= \sum_{i}\Perp{a_i}\Par{a_i}\|\vtilde_i\|^2 - \sum_{i}(1-\lambdatilde_i)\Perp{a_i}\Par{a_i}\|\vtilde_i\|^2
	= 0 - \sum_{i}(1-\lambdatilde_i)\Perp{a_i}\Par{a_i}\|\vtilde_i\|^2. \]
Define vectors $\Par{\bigx}:=\sum_i \sqrt{1-\lambdatilde_i}\Par{a_i}\vtilde_i$ and
$\Perp{\bigx}:= \sum_i \sqrt{1-\lambdatilde_i}\Perp{a_i}\vtilde_i$.  Then because the
$\vtilde_i$ are mutually orthogonal, we have
\[ \langle \Perp{\bigx},\Par{\bigx}\rangle = \sum_i (1-\lambdatilde_i)\Perp{a_i}\Par{a_i}\|\vtilde_i\|^2 
	= - \langle \Perp{\xhat},\Par{\xhat}\A \rangle.\]
By Cauchy-Schwartz, $|\langle \Perp{\bigx},\Par{\bigx}\rangle|\leq
\|\Perp{\bigx}\|\|\Par{\bigx}\|$.  Moreover,
\[ \|\Par{\bigx}\|^2 = \sum_i (1-\lambdatilde_i)(\Par{a_i})^2\|\vtilde_i\|^2 
	= \|\Par{\xhat}\|^2 - \sum_i \lambdatilde_i(\Par{a_i})^2\|\vtilde_i\|^2
	= \|\Par{\xhat}\|^2 - \langle \Par{\xhat}, \Par{\xhat}\Atilde \rangle
	= (1-\deltahpat)\|\Par{\xhat}\|^2. \]
Similarly, $\|\Perp{\bigx}\|^2 = (1-\deltahpet)\|\Perp{\xhat}\|^2$.  Taken
together, this proves the lemma.
\qed\\

\subsection{Deferred proofs for the first decomposition theorem}\label{sec:minimized}

Recall in the proof of Theorem~\ref{thm:main}, we showed
\[
\Gap(\m) \geq \min_{x\perp\sqrt{\pi}:\|x\|=1} 
(1-\deltahpeh)\|\Perp{\xhat}\|^2+\left(\sqrt{1-\deltahpet}\|\Perp{\xhat}\| - \sqrt{1-\deltahpat}\|\Par{\xhat}\|\right)^2.
\]
To complete the proof, we wish to show that the minimum occurs when
$\deltahpet$ is minimized, and
\[
\Gap(\m) \geq \min_{p^2+q^2=1}\gmin q^2+\left(q\sqrt{1-\deltahpet} - p\sqrt{\gammabar}\right)^2.\]

 \begin{proof}
We assume $(1-\deltahpet)\neq 0$, as otherwise the statement is obvious.  Similarly, we assume $\gmin$ and $\gammabar$ are nonzero.
By setting $q=\|\Perp{\xhat}\|$ and $p=\|\Par{\xhat}\|$, we immediately get
\[\Gap(\m) \geq \min_{p^2+q^2=1} 
(1-\deltahpeh)q^2+\left(\sqrt{1-\deltahpet}q - \sqrt{1-\deltahpat}p\right)^2.\]
It is easy to see that the expression on the right is minimized when $(1-\deltahpeh)$ is minimized.  Next we use a bit of calculus to show that the expression is minimized when $(1-\deltahpet)$ is maximized and when $(1-\deltahpat)$ is minimized.

Define $f(\gamma,r,s) = \min_{p^2+q^2=1} \gamma q^2+\left(sq -rp\right)^2$.
 It is easy to show that the minimum over $p$ and $q$ occurs at the
 values $p^*$ and $q^*$ satisfying $p^*/q^*=(-b+\sqrt{b^2+4})/2$,
 where $b=(r^2-s^2-\gamma)/(rs)$.   
The function $f$ is increasing with $r$ when the partial derivative $f_r = -2(sq-rp)p$ is positive; i.e. when $p/q\geq s/r.$  Similarly, $f$ is decreasing with $s$ when the partial derivative $f_s = 2(sq-rp)q$ is negative, again when $p/q\geq s/r.$  Thus it suffices to show that $p^*/q^*\geq s/r$.  Notice $b\leq r/s - s/r$.  This implies $\sqrt{b^2+4}\geq 2s/r+b$, and so $p^*/q^*\geq s/r$.
 \end{proof}

\subsection{Comparison with previous work}\label{sec:compare}
In this section, we compare Theorem~\ref{thm:main} with results
from~\cite{madr,MR00,jstv04}. 

First, we notice $\deltahpet\geq \mu_{\min}$, where $\mu_{\min}$ is
the smallest eigenvalue of $\Atilde$.  To see this, recall that 
$\xhat = \sum_i \Perp{a_i}v_i$, $\xhat \Atilde= \sum_i 
\Perp{a_i}\mu_iv_i$, and therefore $\langle \xhat, \xhat\Atilde\rangle 
=\sum_i \mu_i (\Perp{a_i})^2\|v_i\|^2\geq \mu_{\min}\|\xhat\|^2.$
Thus, if $\m$ is lazy then $\mu_{\min}\geq 0$ and so $\deltahpet\geq
0$.

\begin{Corollary}\label{cor:compareproduct}
  Assume $\m$ is lazy.  Then $\gamma\geq \gmin\gammabar/3$.
\end{Corollary}
In fact, one can show that the constant is $1/2$ if $\gmin,\gammabar\leq 1/2$ or
$1-\deltahpet\leq 1/2$.
\begin{proof}
Since the bound in Theorem~\ref{thm:main} is minimized when 
$\deltahpet$ is minimized, we may assume $\deltahpet=0$.
Thus, $\rho=1/\sqrt{\gammabar}$.  We will show that for all $p,q$
satisfying $p^2+q^2=1$, we have
\[\gmin q^2+\gammabar\left(p - \frac{q}{\sqrt{\gammabar}}\right)^2\geq
\frac{ \gmin\gammabar}{3}.\]
Clearly if $q^2\geq \gammabar/3$, then we are done.  So we may assume
$q^2/\gammabar<1/3$.  Notice that since $\gammabar\leq 1$,
\[\left(p - \frac{q}{\sqrt{\gammabar}}\right)^2
~=~ \left(\sqrt{1-q^2} - \frac{q}{\sqrt{\gammabar}}\right)^2
~\geq~ \left(\sqrt{1-\frac{q^2}{\gammabar}} -
\frac{q}{\sqrt{\gammabar}}\right)^2.\]
As $q^2/\gammabar<1/3$, we have
$\left(\sqrt{1-\frac{q^2}{\gammabar}} -
\frac{q}{\sqrt{\gammabar}}\right)^2\geq \frac{1}{3} -
\frac{q^2}{\gammabar}.$  Therefore,
\begin{align*}
  \gmin q^2+\gammabar\left(p -  \frac{q}{\sqrt{\gammabar}}\right)^2
  &\geq  \gmin q^2+\frac{\gammabar}{3} -q^2\\
  &= \frac{\gammabar}{3} - q^2(1-\gmin)\\
  &> \frac{\gammabar}{3}(1-(1-\gmin)) \\
  &= \gmin\gammabar/3.
\end{align*}
\end{proof}


Next, we will show that a variant of Theorem~1 of~\cite{jstv04} follows from
Theorem~\ref{thm:main} from this paper, which is the content of
Corollary~\ref{cor:compareT}.  Let $T = 
\max_i\max_{\sigma\in\Omegahat_i}\sum_{\tau\in\Omega\setminus\Omegahat_i}\transmat(\sigma,\tau)$.  This is the parameter $\gamma$ from~\cite{jstv04}. We will now show that $\deltahpet\geq 1-2T$.
Notice the probability of a move in $\Ahat$ is at least 
$1-T$, so every element has a self-loop probability of at least $1-T$ 
in $\Atilde$. 
Thus, $\Atilde$ can be 
written as $\Atilde = T\Atilde' + I(1-T)$ for some transition matrix 
$\Atilde'$ with minimum eigenvalue $-1$.  This implies that the 
minimum eigenvalue of $\Atilde$ satisfies $\mu_{\min}\geq 1-2T$.  On 
the other hand, $\deltahpet\geq \mu_{\min}$.  

\begin{Corollary}\label{cor:compareT}
 $\gamma\geq \min\left\{\frac{\gammabar}{3}, \frac{\gmin\gammabar}{3T+\gammabar}\right\}$.
\end{Corollary}
\begin{proof}
We wish to show $\gamma\geq \min\left\{\frac{\gammabar}{3}, \frac{\gmin\gammabar}{3T+\gammabar}\right\}$.
  As $\rho^2\leq 2T/\gammabar$, it suffices to show that for all $p^2+q^2=1$,
 \[\gmin q^2+\gammabar(p-\rho q)^2\geq \min\left\{\frac{\gammabar}{3}, \frac{\gmin}{1+1.5\rho^2}\right\}.\]
 If $q^2\geq \frac{1}{1+1.5\rho^2}$ then we are done, so we may assume
 $\frac{1}{1+1.5\rho^2}-q^2\geq 0$.
 Define
  \[f_1=\frac{(p-\rho q)^2}{1/(1+1.5\rho^2) - q^2} \quad \text{and}
  \quad f_2=\frac{1/3 - (p-\rho q)^2}{q^2}.\]
 Notice that $\gmin q^2+\gammabar(p-\rho q)^2\geq
 \frac{\gmin}{1+1.5\rho^2}$ if and only if $\gmin/\gammabar\leq f_1$.
 On the other hand, $\gmin q^2+\gammabar(p-\rho q)^2\geq \gammabar/3$ if
 and only if $\gmin/\gammabar\geq f_2$.  Thus, it suffices to show
 $f_1\geq f_2$ for all parameter choices.
First, notice that since $\frac{1}{1+1.5\rho^2}-q^2\geq 0$, we have  $f_1\geq f_2$ whenever
 \[1\geq \left(\frac{1}{3(p-\rho
   q)^2}-1\right)\left(\frac{1}{(1+1.5\rho^2)q^2}-1\right).\]
This is satisfied whenever $3(p-\rho q)^2\geq 1-(1+1.5\rho^2)q^2$.
Expanding and bringing all to the left hand side shows this is true because $(2p-3\rho q)^2\geq 0$.
\end{proof}

We do not currently have a comparison between our
Theorem~\ref{thm:main} and Corollary~2 of~\cite{jstv04} that requires
a pointwise bound of $\pi_i^j$.  However, see Remark~\ref{rem:eta}
which shows that their result would not be sufficient for our
application to permutations.

\subsection{Proof of the Complementary Decomposition Theorem}\label{sec:compDec}
First, we prove Theorem~\ref{thm:primaldual}.
Recall that $\lambdatilde_1\geq\lambdatilde_2\geq\ldots\geq 
\lambdatilde_{|\Omega|}$ are the eigenvalues of $\Atilde$ with 
corresponding eigenvectors $\vtilde_1,\vtilde_2,\ldots, 
\vtilde_{|\Omega|}$, and that for any $v \in \R^{|\Omega|}$, we write 
$v=\sum_{i}a_i\vtilde_i$ for some constants $a_1,a_2,\ldots, 
a_{|\Omega|} \in \R$.  Also, $\|v\|^2 = \sum_i a_i^2\|\vtilde_i\|^2,$ 
and $v\Atilde = \sum_i a_i\lambdatilde_i\vtilde_i$. Define the set 
$S=\{i:\mu_i=1\}$.
Let 
$\deltatpet = 
\frac{\langle\Perp{\xtilde},\Perp{\xtilde}\Atilde\rangle}{\|\Perp{\xtilde}\|^2}$.
 Now we can make an explicit statement about the gap 
of $\m$; notice the equality in Equation~\ref{eq:primaldual}.  
\begin{customthm}{\ref{thm:primaldual}}
\begin{equation}\label{eq:primaldual}
\Gap(\m)= \min_{x\perp \sqrt{\pi}, \|x\|=1} (1-\deltahpeh)\|\Perp{\xhat}\|^2+
(1-\deltatpet)\|\Perp{\xtilde}\|^2.
\end{equation}
In particular,
\begin{equation}\label{eq:primaldual2}
\Gap(\m)\geq \min_{x\perp \sqrt{\pi}, \|x\|=1}  
\gminhat\|\Perp{\xhat}\|^2+\gmintilde\|\Perp{\xtilde}\|^2.
\end{equation}
\end{customthm}

\begin{proof}
Notice $\deltatpet\|\Perp{\xtilde}\|^2 = \sum_{i\in S}\lambdatilde_i(\Perp{a_i}+\Par{a_i})^2$.
Thus,
\[(1-\deltatpet)\|\Perp{\xtilde}\|^2 = \sum_{i}(1-\lambdatilde_i)(\Perp{a_i}+\Par{a_i})^2
   = (1-\deltahpet)\|\Perp{\xhat}\|^2 + (1-\deltahpat)\|\Par{\xhat}\|^2 - 
   2\langle \Perp{\xhat},\Par{\xhat}\Atilde\rangle.\]
On the other hand, from Equation~\ref{eq:strongestbound}, we have
\[1-\langle x,x\A\rangle = (1-\deltahpeh)\|\Perp{\xhat}\|^2+(1-\deltahpet)\|\Perp{\xhat}\|^2+
(1-\deltahpat)\|\Par{\xhat}\|^2 - 2\langle \Perp{\xhat},\Par{\xhat}\Atilde\rangle.\]
Thus, for all $x\perp \sqrt{\pi}$ with norm 1, we have
\[1-\langle x,x\A\rangle=(1-\deltahpeh)\|\Perp{\xhat}\|^2+(1-\deltatpet)\|\Perp{\xtilde}\|^2.\]
Applying Equation~\ref{eq:eigen}, we get Equation~\ref{eq:primaldual}.
To get Equation~\ref{eq:primaldual2}, we apply Lemma~\ref{lem:eta}, which yields
$1-\deltahpeh \geq 1-\lambda_{\max} = \gminhat$.  An analogous statement
to Lemma~\ref{lem:eta} holds for $\deltatpet$, and shows 
$1-\deltatpet \geq \gmintilde$. 
\end{proof}

We can now prove the Complementary Decomposition theorem.\\

\noindent\emph{Proof of Theorem~\ref{thm:indep}}.
By Equation~\ref{eq:primaldual2} from Theorem~\ref{thm:primaldual}, 
it suffices to show that $\|\Perp{\xhat}\|^2+ \|\Perp{\xtilde}\|^2\geq
(1-\epsilon)^2$. If $\|\Perp{\xhat}\|^2\geq (1-\epsilon)^2$ we are done,
so we may assume not.
As the vectors within each row of the table are orthogonal, we have
$(1-\epsilon)^2>\|\Perp{\xhat}\|^2 =\|\c\|^2+ \|\d\|^2$.
Furthermore, since
\[0=\langle\Perp{\xhat}, \Par{\xhat}\rangle =\sum_i
\Perp{a_i}\Par{a_i}\|v_i\|^2\]
we have
\[\langle\c,\a\rangle = \sum_{i\notin S}\Perp{a_i}\Par{a_i}\|v_i\|^2 = 
-\sum_{i \in S}\Perp{a_i}\Par{a_i}\|v_i\|^2 = -\langle\d,\b\rangle.\]
Thus,
\begin{align}
    \|\Perp{\xtilde}\|^2 &=\|\c+\a\|^2\nonumber\\
    &=\langle\c+\a,\c+\a\rangle\nonumber\\
    &=\|\c\|^2+\|\a\|^2+2\langle\c,\a\rangle\label{eq:innerabcd}\\
    &=\|\c\|^2+\|\a\|^2-2\langle\d,\b\rangle\nonumber\\
    &\geq\|\c\|^2+\|\a\|^2-2\|\d\|\|\b\|.\nonumber
\end{align}
We used Cauchy-Schwartz for the final inequality.  Putting everything 
together,
\begin{align*}
\|\Perp{\xhat}\|^2+ \|\Perp{\xtilde}\|^2 &\geq\|\c\|^2+ \|\d\|^2 +\|\c\|^2+\|\a\|^2-2\|\d\|\|\b\|\\
  &=\|\c\|^2+ \|\d\|^2 + (1-\|\d\|^2-\|\b\|^2)-2\|\d\|\|\b\|\\
  &= 1 + \|\c\|^2-\|\b\|^2-2\|\d\|\|\b\|\\
  &\geq 1 + \|\c\|^2-\epsilon^2-2\epsilon\sqrt{(1-\epsilon)^2-\|\c\|^2}\\
  &\geq 1 -\epsilon^2-2\epsilon(1-\epsilon)\\
  &= (1 -\epsilon)^2.
\end{align*}
Notice that we have used the fact that $1=\|x\|^2 = \|\b\|^2+\|\a\|^2+\|\d\|^2+\|\c\|^2$.  
Also, we used the bound $\|\c\|^2 \geq 0$.  If $\|\c\|$ is much smaller than $\|\b\|$ then the bound \[\|\Perp{\xhat}\|^2+ \|\Perp{\xtilde}\|^2\geq 1-\epsilon^2 - 2\|\c\|\]
is tighter.  This is obtained by applying Cauchy-Schwartz to $\langle\c,\a\rangle$ in Equation~\ref{eq:innerabcd}.  However, to apply this requires analyzing $\|\c\|$, which appears challenging. \qed\\

\subsection{Bounding $\epsilon$ revisited: the proof of Lemma~\ref{lem:directproduct}}\label{sec:r}
In this section, we show how to bound $\epsilon.$  Recall \[r(i,j)=\frac{\pi(\Omegahat_i\cap\Omegatilde_j)}{\pi(\Omegahat_i)\pi(\Omegatilde_j)}.\]
We now prove Lemma~\ref{lem:directproduct}, which states that
  \[\|\b\|^2\leq \sum_{(i,j)}\pi(\Omegahat_i\cap\Omegatilde_j)(\sqrt{r(i,j)}-1/\sqrt{r(i,j)})^2.\]
\noindent\emph{Proof of Lemma~\ref{lem:directproduct}}.
Recall $\b$ is the projection of $\Par{\xhat}$ onto the top
eigenvectors of $\Atilde$.  The top eigenvectors of $\Atilde$ are precisely
the set of all $\sqrt{\pitilde_j}$ for $j\in[\rtilde]$.
Therefore,
\[\b = \sum_{j}\frac{\langle
\Par{\xhat},\sqrt{\pitilde_j}\rangle}{\|\sqrt{\pitilde_j}\|^2}\sqrt{\pitilde_j}.\]
As the eigenvectors of $\Atilde$ are an orthonormal basis, we have
\[\|\b\|^2 = \sum_{j}\langle \Par{\xhat},\sqrt{\pitilde_j}\rangle^2.\]
For any $j\neq j'\in[\rtilde]$ and any $\sigma\in\Omegatilde_{j'}$, $\sqrt{\pitilde_j(\sigma)}=0$  and for $i\in[\rhat]$, $\pitilde_j(\Omegahat_i\cap\Omegatilde_j) = \pi(\Omegahat_i\cap\Omegatilde_j)/\pi(\Omegatilde_j)$.  Therefore,
\begin{equation}\label{eq:innerprod}
 \langle \Par{\xhat},\sqrt{\pitilde_j}\rangle 
 ~= \sum_i\sum_{\sigma\in\Omegahat_i}\Par{\xhat}(\sigma)\sqrt{\pitilde_j(\sigma)}
 ~= \sum_i\alpha_i\sum_{\sigma\in\Omegahat_i\cap\Omegatilde_j}\sqrt{\pihat_i(\sigma)\pitilde_j(\sigma)}
 ~= \sum_{i\in[\rhat]}\alpha_i\frac{\pi(\Omegahat_i\cap\Omegatilde_j)}{\sqrt{\pi(\Omegahat_i)\pi(\Omegatilde_j)}}.
 \end{equation}



Since $x \perp \sqrt{\pi}$ and $\Perp{\xhat} \perp \sqrt{\pi}$
by definition, it follows that $\Par{\xhat} \perp \sqrt{\pi}$ as well.
This implies that $\alpha \perp \sqrt{\pibar}$, as 
\begin{equation}\label{eq:alphaperp}
0 = \langle \Par{\xhat},\sqrt{\pi} \rangle
  = \sum_i \alpha_i \sum_{\elemone \in \Omegahat_i} \sqrt{\pihat_i(\elemone)\pi(\elemone)}
  = \sum_i \alpha_i \sum_{\elemone \in \Omegahat_i} \frac{\sqrt{\pi(\elemone)}}{\sqrt{\pi(\Omegahat_i)}} \sqrt{\pi(\elemone)} 
  = \sum_i \alpha_i \sqrt{\pi(\Omegahat_i)},
\end{equation}
and this final term is equal to $\sum_i \alpha_i \sqrt{\pibar_i}=\langle \alpha, \sqrt{\pibar} \rangle.$
Multiplying Equation~\ref{eq:alphaperp} by $\pi(\Omegatilde_j)$ and subtracting it from Equation~\ref{eq:innerprod}, we have
\[
 \langle \Par{\xhat},\sqrt{\pitilde_j}\rangle 
 =\sum_{i\in[\rhat]}\alpha_i\left(\frac{\pi(\Omegahat_i\cap\Omegatilde_j)}{\sqrt{\pi(\Omegahat_i)\pi(\Omegatilde_j)}}-
    \sqrt{\pi(\Omegahat_i)\pi(\Omegatilde_j)}\right)
 =\langle\alpha,V_j\rangle,
 \]
where 
\[V_j(i):=\left(\frac{\pi(\Omegahat_i\cap\Omegatilde_j)}{\sqrt{\pi(\Omegahat_i)\pi(\Omegatilde_j)}}-
\sqrt{\pi(\Omegahat_i)\pi(\Omegatilde_j)}\right)=\sqrt{\pi(\Omegahat_i\cap\Omegatilde_j)}(\sqrt{r(i,j)}-1/\sqrt{r(i,j)}).\]  
By the Cauchy-Schwartz inequality and
the fact that $\|\alpha\|=\|\Par{\xhat}\| \leq \|x\|=1$, we have
$ \langle\alpha,V_j\rangle\leq \|\alpha\|\|V_j\| = \|V_j\|.$
Therefore we get,
\[ \|\b\|^2 = \sum_{j}\langle \Par{\xhat},\sqrt{\pitilde_j}\rangle^2
	\leq \sum_j\|V_j\|^2 
	= \sum_{i,j}\pi(\Omegahat_i\cap\Omegatilde_j)(\sqrt{r(i,j)}-1/\sqrt{r(i,j)})^2.\]
\qed\\

\section{Additional Details for the Biased Permutations Application}\label{app:perm}
In this section, we include the remaining details of the proof of Theorem~\ref{pprocess}, which gives a bound on the spectral gap of the particle process chain $\mk,$ and shows how we can use this to bound the mixing time of the nearest-neighbor chain $\mn.$  The section is laid out as follows.  We begin in Section~\ref{sec:MKprelim} by giving some Markov chain preliminaries, describing the stationary distribution of $\mk,$ and formally giving the definition of the weakly monotonic property.  Next, in Section~\ref{sec:mk} we give the complete proofs of Theorem~\ref{pprocess} and Lemma~\ref{lem:dual} which give a bound on the spectral gap of $\mk$.  In Section~\ref{sec:ortho} we complete the proof of Lemma~\ref{lem:perms} which says that our decomposition is $1/n$-orthogonal. Finally, in Section~\ref{sec:mn} we formally define the nearest neighbor chain $\mn$ and show how we can use the bound on the spectral gap of $\mk$ to bound the mixing time of $\mn.$   

\subsection{Markov chain and $\mk$ preliminaries}\label{sec:MKprelim}
We begin with some preliminaries on Markov chains and mixing times.  The time a Markov chain takes to converge to its stationary distribution, or \emph{mixing time}, is measured in terms of the distance between the distribution at time~$t$ and the stationary distribution. The \emph{total variation distance} at time~$t$ is
$\|P^t,\pi\| _{tv} = \max_{x\in\Omega}\frac{1}{2}\sum_{y\in\Omega} |P^t(x,y)-\pi(y)|,$
where $P^t(x,y)$ is the $t$-step transition probability.  For all $\epsilon>0$, the \emph{mixing time} $\tau(\epsilon)$ of $\m$ is defined as
$$\tau(\epsilon)=\min \{t: \|P^{t'},\pi \|_{tv}\leq \epsilon, \forall t' \geq t\}.$$ 

In order to use the lower bound on the spectral gap to obtain an upper bound on the mixing time we will use the following well-known result. 
Again, let $\Gap(\transmat)=1-|\lambda_1|$ denote the spectral gap, where $\lambda_0, \lambda_1, \ldots, \lambda_{|\Omega|-1}$ are the eigenvalues of the transition matrix $\transmat$ and $1=\lambda_0>|\lambda_1|\geq |\lambda_i|$ for all $i\geq 2$. The following result  relates the spectral gap with the mixing time (see, e.g., \cite{sinclair},\cite{RM03}):

\begin{theorem}\cite{RM03}\label{gap}
Let $\pi_*=\min_{x\in\Omega} \pi (x)$. For all $\epsilon>0$ we have
\begin{align*}
(a)\qquad \tau(\epsilon)&\leq \frac{1}{1-|\lambda_1|} \log\left(\frac{1}{\pi_*\epsilon}\right).\\
(b)\qquad \tau(\epsilon)&\geq \frac{|\lambda_1|}{2(1-|\lambda_1|)} \log\left(\frac{1}{2\epsilon}\right).
\end{align*}
\end{theorem}

The chain $\mk$ connects the state space $\Omega$ and has the stationary distribution (see e.g.,~\cite{bmrs})
$$\pi(\sigma) = \prod_{i < j: \sigma(i) > \sigma(j)}\frac{p_{\sigma(i),\sigma(j)}}{p_{\sigma(j),\sigma(i)}} Z^{-1} = \prod_{i < j: \sigma(i) > \sigma(j)}q_{\sigma(i), \sigma(j)} Z^{-1} $$ where $Z$ is the normalizing constant $\sum_{\sigma \in \Omega} \pi(\sigma)$ and $q_{\sigma(i), \sigma(j)}=\frac{p_{\sigma(i),\sigma(j)}}{p_{\sigma(j),\sigma(i)}}.$  
In our analysis of $\mk$ (and of $\mn$) we require the weakly monotonic condition given in~\cite{bmrs,MS}: 
 
\begin{Definition}[\cite{bmrs}]\label{mono}The set $\p$ is weakly monotonic if properties 1 and either 2 or 3 are satisfied.
\begin{enumerate}
\item $p_{i,j} \geq 1/2$ for all $2 \leq i < j \leq n,$ and
\item $p_{i,j+1} \geq p_{i,j}$ for all $1 \leq i < j \leq n-1$ or
\item $p_{i-1,j} \geq p_{i,j}$ for all $2 \leq i < j \leq n.$
\end{enumerate}
\end{Definition}
\noindent As in~\cite{MS}, we will assume that property (2) holds.  If
instead property (3) holds, then as described in~\cite{MS} we would
modify $\mk$ (and $\mtk$ defined below) to allow swaps between
elements in different particle classes across elements whose particle
classes are larger (instead of smaller) and modify the induction so
that at each step $\sigma_{i}$ restricts the location of particles
larger than $i$ (instead of smaller).

\subsection{Proof of Theorem~\ref{pprocess} and Lemma~\ref{lem:dual}}\label{sec:mk}
Next we include a complete proof of Theorem~\ref{pprocess} which bounds the mixing time of the particle process chain $\mk.$  As described in Section~\ref{sec:perm} we will do this using induction and our new complementary decomposition theorem.  Recall that at each step our restriction chain will be $\m_{\sigma_i}$ which is rapidly mixing by induction.  We begin by proving Lemma~\ref{lem:dual} which states that each complementary restriction chain is a bounded generalized biased exclusion process with spectral gap~$\Omega(n^{-2}).$\\  

\noindent\emph{Proof of Lemma~\ref{lem:dual}.}
Moves of complementary restrictions $\Ptilde_1, \Ptilde_2, \ldots, \Ptilde_{\rtilde}$  involve exchanging an element
from $\C_i$ with an element from $\C_j$ where $j >i$.  There may be
additional elements between the elements being exchanged but if there
are, they are in a smaller particle class $\C_s$ with $s < i.$  If we
view all elements in $\C_i$ as one type and all elements in $\C_{i+1}, \C_{i+2},
\ldots, \C_{k} $ as another, then each complementary restriction chain  can be viewed as a generalized exclusion process.  We use the following result.

  \begin{theorem}\cite{MS}\label{boundedBias} Let $\me$ be a bounded generalized exclusion process on $n_1$ 1's and $n_0$ 0's.  Suppose without loss of generality that $n_1 \leq n_0,$  then the spectral gap of $\me$ is $$\Omega(1/\left(\left(n_0+n_1\right)\left(n_1+\ln n_0)\right)\right).$$
\end{theorem}

\noindent Since the probabilities $\transmat$ are weakly monotonic, the exclusion process involving the $i$ particles and the particles larger than $i$ is bounded and we can
apply Theorem~\ref{boundedBias}.  Assume $|C_i|=c_i$ for all $i$.  There are $c_i$ particles of type
$i$, $\sum_{j=i+1}^k c_j < n$ particles of the other type, and the
moves are selected with probability $c_i/(4n).$  Applying
Theorem~\ref{boundedBias} shows that the
spectral gap of each complementary restriction chain is $\Omega(\frac{c_{i}}{n}\frac{1}{n(c_{i} + \ln n)})=\Omega(n^{-2}),$ since $c_{i}\geq \ln n$.  

\qed

\noindent Next, we prove Theorem~\ref{pprocess} which says that if the probabilities $\transmat$ are weakly monotonic, form a bounded $k$-class for $k\geq 2,$ and  $|C_i|\geq 2\psize$ for $1\leq i \leq k$ then the spectral gap 
of the chain $\mk$ satisfies $\text{Gap}(\transmat) =\Omega(n^{-2}).$\\

\noindent\emph{Proof of Theorem~\ref{pprocess}.}
  At each step of
 the induction, we apply complementary decomposition (Theorem~\ref{thm:indep}).
 The restrictions of each decomposition will be rapidly mixing by
 induction and the complementary restriction chains are bounded generalized exclusion
 processes, by Lemma~\ref{lem:dual}.  The base case is $i=k-2$ and the final decomposition is $i=0$.

We begin with our base case, $i=k-2$.  
Let $\sigma_{k-2}$ be any fixed location of the particles in classes
$\C_1, \ldots, \C_{k-2}$.  
The Markov chain $\m_{\sigma_{k-2}}$ rejects all moves of $\mk$ unless they
exchange a particle in class $\C_{k-1}$ with a particle in class
$\C_{k}$.  Thus, its moves only involve two types of particles, with
all other particles fixed, so we can view 
$\m_{\sigma_{k-2}}$ as a generalized exclusion process.  
Thus,
$\m_{\sigma_{k-2}}$ is a bounded generalized exclusion process slowed down by a factor of $c_{k-1}/(4n)$ (as this is the probability that these transitions are chosen).  By Theorem~\ref{boundedBias} for any such $\sigma_{k-2}$, $\m_{\sigma_{k-2}}$ has spectral gap $\Omega(\frac{c_{k-1}}{n}\frac{1}{n(c_{k-1} + \ln n)})=\Omega(n^{-2}),$ since $c_{k-1}\geq \ln n$.  

We assume by induction the mixing time bound holds for all
$\m_{\sigma_i}$ for some $i\leq k-2$, and we
will use this result to prove that our mixing 
time bound holds for all $\m_{\sigma_{i-1}}$.
Let $\sigma_{i-1}$ represent any fixed choice of locations for all
elements in classes $C_1,C_2,\ldots, C_{i-1}$.  In order to bound
$\text{Gap}(\transmat_{i-1})$ we will apply Theorem~\ref{thm:indep}.  Given any
$\sigma_i$ that is consistent with $\sigma_{i-1}$ (i.e. they agree on
the locations of all elements in classes $C_1,C_2,\ldots, C_{i-1}$),
the Markov chain $\m_{\sigma_i}$ will be a restriction Markov chain of
$\m_{\sigma_{i-1}}$.  By induction, we have $\text{Gap}(\m_{\sigma_i})=
\Omega(n^{-2}\left(1-\frac{1}{n}\right)^{2(k-2-i)}).$ 
Lemma~\ref{lem:dual} shows that the complementary restrictions have minimum gap $\Omega(n^{-2}).$ 
By Lemma~\ref{lem:perms}, we know that the decomposition is $1/n$-orthogonal.
Combining these with Theorem~\ref{thm:indep} implies 
$\text{Gap}(\m_{\sigma_{i-1}})= \Omega\left(n^{-2}\left(1-\frac{1}{n}\right)^{2(k-2-i)+2}\right).$ Substituting $i = 0$ gives the desired theorem $$\text{Gap}(\mk) = \Omega\left(n^{-2}\left(1-\frac{1}{n}\right)^{2k-2}\right) = \Omega\left(n^{-2}\left(1-\frac{1}{n}\right)^{n}\right) = \Omega\left(n^{-2}\right).$$ 

\qed

\subsection{The decomposition is $1/n$-orthogonal.}\label{sec:ortho}  Here we include a complete proof of Lemma~\ref{lem:perms} which states that we have a $1/n$-orthogonal decomposition at each step of the induction.  Many of the details are described in Section~\ref{sec:perm} but are also included here for ease of reading.  

Recall that at each step of the induction $a$ is a 2-particle system obtained from $\sigma_i$ by removing all particles of type less than $i$ with $|C_i|$ particles of type $i$ and $\sum^{k}_{j=i+1}|C_j|$ particles of type `$\_$'.  Also, $b$ is the permutation of the elements of type bigger than $i$ and the location of the particles smaller than $i-1$ is always fixed (see Figure~\ref{fig:product}).   
Recall that $a$ is ``good" if it has fewer than $N^*=C_q\log n$ inversions and  $b$ is ``good" if it has fewer than $N^*$ inversions involving particles
of type $i+1$.  As there are at least $2N^*$ particles of type $i+1$, if $a$ and $b$ are both good, then $(a,b)$ has no inversions between $i$ and $j$ for $j>i+1$.  We want to show that almost all weight contributing to $\pi$ comes from pairs $(a,b)$ where both $a$ and $b$ are good.  We will do this by using a bound on 2-particle systems with more than $N^*$ inversions (i.e. staircase walks with area larger than $N^*$ under the curve.
The number of staircase walks with area $\area$ under the
  curve is precisely the partition number $p(\area).$  Hence, we will use the following theorem of Erd\H{o}s.

\begin{theorem}\cite{Erdos42}\label{thm:HR}
  The number $p(\area)$ of partitions of an integer $\area$ satisfies $p(\area) < e^{\pi\sqrt{2\area/3}}.$ 
\end{theorem}

\noindent Next, we will use this bound on the partition number to bound a sum of exclusion processes where each inversion contributes a factor of $q$ to the  weight of the exclusion process.   Define $c_1=e^{\pi\sqrt{2/3}}\approx 13.$  Theorem~\ref{thm:HR}
implies $p(\area) < c_1^{\sqrt{\area}}.$  
Suppose $q<1$ is a constant 
and let $\psize = \max\{\frac{\log (6n^2)+\log(1+2/(q^{-1}-1))}{\log(1+q^{-1})-1}, \frac{\log^2(c_1)}{\log(2/(q(1+q^{-1})))}\} = \Theta(\log n).$  Let $N(\sigma)$ be the number of inversions in a particular exclusion process $\sigma$ then we prove the following.

\begin{Lemma}\label{lem:partitions}
  For any biased exclusion process with minimum bias constant $q < 1$, we have
\[\sum_{\sigma:\area(\sigma)\geq \psize}q^{\area(\sigma)} \leq 1/(6n^2).\]
\end{Lemma}
\begin{proof}
  The number of staircase walks with area $\area$ under the
  curve is precisely the partition number $p(\area).$  For all $\area\geq
  \frac{\log^2(c_1)}{\log(2/(q(1+q^{-1})))}$, we have
  $c_1^{\sqrt{\area}}q^\area\leq (2/(1+q^{-1}))^\area.$
  Therefore
  
  \begin{align*}
  \sum_{\sigma:\area(\sigma)\geq \psize}q^{-\area(\sigma)}&\leq \sum_{\area\geq
    \psize}p(\area)q^{\area}\\
  &\leq \sum_{\area\geq \psize}c_1^{\sqrt{\area}}q^{\area}\\
  &\leq \sum_{\area\geq \psize}(2/(1+q^{-1}))^{\area}\\
  &=\left(\frac{2}{1+q^{-1}}\right)^{\psize}\frac{1}{1-\frac{2}{1+q^{-1}}}\\
  &=\left(\frac{2}{1+q^{-1}}\right)^{\psize}\frac{1+q^{-1}}{q^{-1}-1}\\
  &\leq 1/(6n^2),
  \end{align*}
  by the choice of $\psize$.  
 \end{proof}

\noindent We are now ready to use Lemma~\ref{lem:partitions} and Lemma~\ref{lem:directproduct} to prove that  our decomposition is $1/n$-orthogonal.\\

\noindent\emph{Proof of Lemma~\ref{lem:perms}.}
Assume that at this step in the induction, all particles in $C_1, C_2, \ldots, C_{i-1}$ are fixed in the same position in all permutations according to some $\sigma_{i-1}.$  The stationary distribution of the chain $\m_{\sigma_{i-1}}$ is $$\pi_{\sigma_{i-1}}(\sigma) =  \prod_{i < j: \sigma(i) > \sigma(j)}q_{\sigma(i), \sigma(j)}Z_{\sigma_{i-1}}^{-1}.$$  
where $Z_{\sigma_{i-1}}$ is the normalizing constant $\sum_{\sigma \in \Omega_{\sigma_{i-1}}} \pi(\sigma),$ the set $\Omega_{\sigma_{i-1}}$ contains the permutations consistent with $\sigma_{i-1},$ and $q_{\sigma(i), \sigma(j)}=\frac{p_{\sigma(j),\sigma(i)}}{p_{\sigma(j),\sigma(i)}}.$  For ease of notation, throughout the remainder of this section we will let $\pi = \pi_{\sigma_{i-1}}$ and $Z =Z_{\sigma_{i-1}}$.

Consider a permutation $\sigma$.  We will essentially break it into two parts $a$ and $b.$  Recall that all particles in classes $i-1$ and smaller will always be in the same position in all permutations at this level of the induction and fixed according to the particular $\sigma_{i-1}.$  Each $a$ corresponds to a 2-particle system on the $i$ particles with higher particles while each $b$ corresponds to a permutation on the remaining $i+1$ and higher particles.  Let $A$ be the set of all possible choices of $a$; $A$ consists of all 2-particle systems with $|C_i|$ particles of type $i$ and $\sum_{j=i+1}^k|C_j|$ particles of type $\_$ (these are in bijection with staircase walks as in Figure~\ref{fig:staircase}).  Define $b$ to be the permutation of the elements of type bigger than $i$ that is consistent with $\sigma$; let $B$ be the set of all possible $b$.  For example, let $\sigma_{2} = 12\_1\_2\_1\_\_,$ where the $\_$ represents an empty location that can be filled with particles in class $C_3$ or higher.  If $a$ is configuration $3\_3\_\_$ (here $\_$ can be filled with particles in class $C_4$ or higher) and $b$ is $454$ then $(a,b) = 1231423154.$  The set of permutations which all have the $i$ particles in exactly the same location as given by $a$ is called $\Omegahat_a$ is a part of the partition corresponding to a restriction chain.  Similarly, the set of permutations all of which have the $i+1$ and higher particles in the same relative order given by $b$ is called $\Omegatilde_b$, which is the state space of a complementary restriction chain.  Let $a^*$ and $b^*$  be the highest weight elements in $A$ and $B$, respectively.  Using these definitions, the permutation $(a^*,b^*)$ has the particles in classes $i-1$ and smaller fixed according to $\sigma_{i-1}$ and all other higher particles in sorted order.   In our example, $b^* = 445,$ $(a,b^*) = 1231423145,$ $a^* = 33\_\_\_,$ $(a^*,b) = 1231324154,$ and $(a^*,b^*) = 1231324145.$

Next, we will decompose the product over inversions in Equation~\ref{eq:statdist} into several quantities.
Notice that $\pi(a^*,b^*)$ is the product over inversions that are in every $\sigma\in\Omega$, normalized by $Z_{\sigma_{i-1}}$.  Define $w(a) = \pi(a,b^*)/\pi(a^*,b^*).$  This is the product over inversions that are in every $\sigma\in\Omegahat_a$ that are not in every $\sigma\in\Omega$. Similarly, $w(b) = \pi(a^*,b)/\pi(a^*,b^*)$ is the product over inversions in every $\sigma\in\Omegatilde_b$, and $w(a,b) = \pi(a,b) \pi(a^*,b^*)/ (\pi(a,b^*)\pi(a^*,b))$ is the product over inversions in $\sigma=(a,b)$ beyond those that are required by being in $\Omegahat_a$ and $\Omegatilde_b$.  From these definitions, it is clear that $\pi(a,b) = w(a)w(b) w(a,b)\pi(a^*,b^*).$  We will prove in the following lemma that if both $a$ and $b$ are good then $w(a,b) = 1;$ i.e. that the weight of $(a,b)$ is determined entirely by being in $\Omegahat_a$ and $\Omegatilde_b$.
\begin{Lemma}\label{lem:good}
If $a$ and $b$ are both good then $w(a,b) = 1.$
\end{Lemma}
\begin{proof}  From the definitions, it is sufficient to show that $\pi(a,b) = \pi(a,b^*)\pi(b,a^*)/\pi(a^*,b^*).$ First consider the inversions $(j,l)$ for $j<i,l<i.$  These are present in every term ($\pi(a,b),$ $\pi(a,b^*),$ $\pi(a^*,b),$ and $\pi(a^*,b^*)$) and thus cancel. The inversions $(i+1,i)$ are exactly the same in $\pi(a,b)$ and $\pi(a,b^*),$ and there are none in $\pi(a^*,b),$ and $\pi(a^*,b^*);$ thus these also cancel.  Similarly inversions $(j,l)$ for  $j\geq i+1,l \geq i$ are exactly the same in $\pi(a,b)$ and $\pi(a^*,b),$ and there are none in $\pi(a,b^*),$ and $\pi(a^*,b^*);$ thus these cancel. Next consider the inversions $(j,l)$ for  $j>i+1,l < i$, since there are no $(j,i)$ inversions  these are the same in $\pi(a,b)$ and $\pi(a^*,b)$ and the same in $\pi(a^*,b^*)$ and $\pi(a, b^*).$ Similarly, the $(i,j)$ for $j < i$ inversions are the same in $\pi(a^*,b)$ and $\pi(a^*,b^*)$ and the same in $\pi(a,b)$ and $\pi(a,b^*).$  Finally the $(i+1,j)$ for $j<i$ inversions are the most complicated.   Assume the $i+1$ particles are numbered and in order in all 4 permutations.  Consider any particular $i+1$ particle and look at its position in $\pi(a,b)$.  If it is to the left of any $i$ particles (i.e. it is involved in an $(i+1,i)$ inversion) then it must come before all particles higher than $i+1$ and it will be in the same position in $\pi(a,b)$ and $\pi(a,b^*)$ and the same positions (which may be different positions) in $\pi(a^*,b)$ and $\pi(a^*,b^*)$  thus any $(i+1,j)$ inversions will cancel.  If not, then it comes after all $i$ particles and then it will be in the same position in $\pi(a,b)$ and $\pi(a^*,b)$ and again the same positions in $\pi(a,b^*)$ and $\pi(a^*,b^*).$  Since there are no $(j,i)$ inversions for $j > i+1$ we can conclude that $w(a,b) = 1$ when $a$ and $b$ are both good.  
\end{proof}

Next, define $\ztilde =  \sum_{a} w(a),$  $\zhat =\sum_{b}w(b),$  and let $\ztilde' = \sum_{a \text{ good}} w(a)$ be the sum $\ztilde$ restricted to only good configurations $a$ and $\ztilde'' =  \sum_{a~\text{bad}} w(a)$ be the sum $\ztilde$ restricted to only bad configurations $a.$  Similarly, let $\zhat' = \sum_{b \text{ good}} w(b)$ and $\zhat'' =  \sum_{b \text{ bad}} w(b).$
Define $\epsilon_1= 1/(6n^2).$  By Lemma~\ref{lem:partitions}, this is a bound on the total weight of staircase walks with area more than $N^*$ under the curve.  We will use this to bound the weight of the bad $a$'s contributing to $\ztilde$ and the weight of the bad $b$'s contributing to $\zhat$.  Specifically, we will prove Lemma~\ref{lem:zBound}, which states the following.
  \begin{enumerate}
  \item $\ztilde'' =  \sum_{a \text{ bad}} w(a) \leq \epsilon_1$ and $\zhat'' =  \sum_{b \text{ bad}} w(b) \leq \epsilon_1\zhat.$
 \item For all $a,b$ we have $\pi(a)\leq w(a)\pi(a^*,b^*)\zhat$ and $\pi(b)\leq w(b)\pi(a^*,b^*)\ztilde.$
 \item For \emph{good} $a$ and \emph{good} $b$ we have $\pi(a) \geq w(a) \pi(a^*,b^*)\zhat'$ and $\pi(b) \geq w(b) \pi(a^*,b^*)\ztilde'.$
 \item For all $b$ we have $\pi(b) \geq w(b)\pi(a^*,b^*)w(a^*,b).$
    \end{enumerate}
\vspace{1em}
\noindent\emph{Proof of Lemma~\ref{lem:zBound}.}  
We begin with the proof of (1).  Recall that $w(a) = \pi(a,b^*)/\pi(a^*,b^*).$  Both permutations $(a,b^*)$ and $(a^*,b^*)$ have the same inversions $(j,l)$ for $j,l < i$ so when we consider the ratio $w(a)$, this contains the $(i,j)$ for $j > i$ inversions and the $(j,l)$ for $j \geq i, l < i$ inversions. Since there are no $(i,j)$ for $j > i$ inversions in $(a^*,b^*)$ and the weight of the $(j,l)$ for $j \geq i, l < i$ inversions is less in $(a,b^*)$ than it is in $(a^*,b^*)$ we have 
$$\ztilde'' =  \sum_{a \text{ bad}} w(a) \leq \sum_{a \text{ bad}}\prod_{j<l: a(j) > i \atop a(l) = i}q_{a(j), i} \leq \sum_{a \text{ bad}}\prod_{j<l: a(j) > i \atop a(l) = i}q_{i+1, i} \leq \epsilon_1.$$
The last two steps follow from $\p$ being weakly monotonic and Lemma~\ref{lem:partitions} respectively.   

Next, we consider $\zhat''$ and recall that $w(b) = \pi(a^*,b)/\pi(a^*,b^*).$  The two permutations $(a^*,b)$ and $(a^*,b^*)$ have the same inversions $(j,l)$ for $j,l \leq i.$  When considering the ratio $w(b)$ there are several types of inversions remaining.  There are the $(j,l)$ for $j,l > i+1$ inversions which are due to the order of the particles $i+2$ and higher which we will represent as $\tau_{i+2}.$  Additionally there are the inversions $(j,i+1)$ for $j > i+1$ of which there are at least $\psize$ if $b$ is \emph{bad} and the $(j,l)$ for $j \geq i+1, l \leq i$ inversions which are maximized in $\tau_{i+2}^*$ which we will define as the highest weight configurations consistent with $\sigma_{i-1}$ and $\tau_{i+2}$.  In other words $\tau_{i+2}^*$ is the configuration that has the particles smaller than $i$ ordered according to $\sigma_{i-1}$, the particles $i$ as far forward as possible, then the particles $i+1$ again as far forward as possible and finally the particles greater than $i+1$ in the remaining positions ordered according to the permutation $\tau_{i+2}.$  Given these definitions, we have the following.
\begin{eqnarray*}
\zhat'' &=&  \sum_{b \text{ bad}} w(b)=  \sum_{\tau_{i+2}} w (\tau_{i+2}^*) \sum_{EX(i+1,j: j>i+1, \text{ bad})}w(b)/w(\tau_{i+2}^*) \\
&\leq& \sum_{\tau_{i+2} }w (\tau_{i+2}^*) \sum_{EX(i+1,j: j>i+1, \text{ bad})}\prod_{j<l: a(j) > i +1\atop a(l) = i+1}q_{a(j), i+1} \\
&\leq& \sum_{\tau_{i+2}} w (\tau_{i+2}^*) \sum_{EX(i+1,j: j>i+1, \text{ bad})}\prod_{j<l: a(j) > i +1\atop a(l) = i+1}q_{i+2, i+1} \\
&\leq& \sum_{\tau_{i+2}} w (\tau_{i+2}^*)\epsilon_1 \leq \epsilon_1 \zhat,
\end{eqnarray*}
where $EX(i+1,j: j>i+1, \text{ bad})$ is an exclusion process (2-particle system) with more than $N^*$ inversions.  One particle is the $i+1$ particles and the other is the particles greater than $i+1$ with more than $N^*$ inversions.  In other words we are dividing the permutation $b$ based on the location of the $i+1$ particles.

Next, we will prove (2) which gives an upper bound on $\pi(a)$ and $\pi(b)$ for all $a,b.$ 
Both of these bounds are straightforward from the definitions. 
For all $a$ we have the following.
 $$\pi(a) = \sum_{b'} \pi(a,b') = \sum_{b'} w(a)w(b')w(a,b')\pi(a^*,b^*) \leq w(a)\pi(a^*,b^*)\sum_{b'}w(b')= w(a)\pi(a^*,b^*)\zhat.$$
Similarly, for all $b$ we have the following.
 $$\pi(b) = \sum_{a'} \pi(b,a') = \sum_{a'} w(a')w(b)w(a',b)\pi(a^*,b^*) \leq w(b)\pi(a^*,b^*)\sum_{a'}w(a')= w(b)\pi(a^*,b^*)\ztilde.$$

\noindent Next, we will prove (3) which gives a lower bound on $\pi(a)$ and $\pi(b)$ for good $a$ and good $b.$  Using the property that if $a$ and $b$ are both good $w(a,b) = 1,$ we have the following bound for \emph{good} $a.$ 
 \begin{eqnarray*}
\pi(a) &=&  \sum_{b'}\pi(a,b') = \sum_{b'} w(a)w(b') w(a,b')\\
&\geq& w(a) \pi(a^*,b^*) \sum_{b' \text{ good}} w(b')w(a,b')\\
&=& w(a) \pi(a^*,b^*) \sum_{b' \text{ good}} w(b') = w(a) \pi(a^*,b^*)\zhat'.\\
\end{eqnarray*}
\noindent For \emph{good} $b$ we have a similar lower bound. 
\begin{eqnarray*}
\pi(b) &=&  \sum_{a'}\pi(a',b) = \sum_{a'} w(a')w(b) w(a',b)\pi(a^*,b^*)\\
&\geq& w(b) \pi(a^*,b^*) \sum_{a' \text{ good}} w(a')w(a',b)\\
&=& w(b) \pi(a^*,b^*) \sum_{a' \text{ good}} w(a') = w(b) \pi(a^*,b^*)\ztilde'.\\
\end{eqnarray*}

Finally, we will prove (4) which gives a weaker lower bound on $b$ that holds for all $b.$
\begin{eqnarray*}
\pi(b) &=& \sum_{a'}\pi(a',b) = \sum_{a'} w(a')w(b) w(a',b)\pi(a^*,b^*)\\
&=& w(b)\pi(a^*,b^*) \sum_a w(a)w(a,b)\\
&\geq& w(b)\pi(a^*,b^*) w(a^*)w(a^*,b) = w(b)\pi(a^*,b^*)w(a^*,b)
\end{eqnarray*}
\qed

Next, we will assume both $a$ and $b$ are good and bound the contribution to $\|\b\|^2.$ Recall that in this case $w(a,b) = 1.$ Using this property  and Lemma~\ref{lem:zBound} part (3) and then (1), we have
\begin{align*}
  r(a,b)&=\frac{\pi(a,b)}{\pi(a)\pi(b)}=\frac{w(a)w(b)\pi(a^*,b^*)}{\pi(a)\pi(b)}\\
  &\leq \frac{w(a)w(b)\pi(a^*,b^*)}{(w(a)\pi(a^*,b^*)\zhat')(w(b) \pi(a^*b^*)\ztilde')}\\
  &= \frac{1}{\pi(a^*,b^*)\zhat'\ztilde'}\leq \left(\frac{1}{(1-\epsilon_1)^2}\right) \frac{1}{\zhat\ztilde \pi(a^*,b^*)} \leq \frac{1}{(1-\epsilon_1)^2}.
\end{align*}
The last step uses the following lower bound on $\zhat\ztilde\pi(a^*,b^*).$
 $$\zhat\ztilde\pi(a^*,b^*) =\sum_{(a,b)} w(a)w(b) \pi(a^*,b^*) > \sum_{(a,b)} w(a)w(b) w(a,b)\pi(a^*,b^*)  = \sum_{(a,b)} \pi(a,b) = 1.$$ 

Using Lemma~\ref{lem:zBound} part (2) we have the following lower bound when $a$ and $b$ are both good,
\begin{align*}
  r(a,b)&=\frac{\pi(a,b)}{\pi(a)\pi(b)}=\frac{w(a)w(b)\pi(a^*,b^*)}{\pi(a)\pi(b)}\\
  &\geq \frac{w(a)w(b)\pi(a^*,b^*)}{
     \left(w(a)\pi(a^*,b^*)\zhat\right)\left(w(b)\pi(a^*,b^*)\ztilde\right)}\\
  &= \frac{1}{
     \zhat\ztilde\pi(a^*,b^*)} \geq(1-\epsilon_1)^2.
\end{align*}
In the last step we upper bound $\zhat\ztilde\pi(a^*,b^*)$ as follows using Lemma~\ref{lem:zBound} part (1).
\begin{equation}\label{prodBound}\zhat\ztilde\pi(a^*,b^*) \leq\frac{\zhat'\ztilde'\pi(a^*,b^*)}{ (1-\epsilon_1)^{2}}= \sum_{a~\text{good},\atop b~\text{good}} \frac{w(a)w(b) \pi(a^*,b^*)}{(1-\epsilon_1)^{2}}< \sum_{(a,b)}\frac{ \pi(a,b)}{(1-\epsilon_1)^{2}} = \frac{1}{(1-\epsilon_1)^{2}}.
\end{equation}

In order to apply Lemma~\ref{lem:directproduct}, we need to bound the quantity $\sum_{(a,b)}\pi(a,b)\left(\sqrt{r(a,b)}-\frac{1}{\sqrt{r(a,b)}}\right)^2.$  There are two cases depending on whether $r(a,b)\leq 1$, but either
way we have
\[\left(\sqrt{r(a,b)}-\frac{1}{\sqrt{r(a,b)}}\right)^2\leq
  \left(\frac{1}{1-\epsilon_1} - (1-\epsilon_1)\right)^2\leq
  5\epsilon_1^2,\]
  as long as $\epsilon_1\leq .191$ (this is true, since
  $\epsilon_1\leq 1/(6n^2)$ and $n\geq 2$.  Therefore
  \[\sum_{a~\text{good},\atop
    b~\text{good}}\pi(a,b)\left(\sqrt{r(a,b)}-\frac{1}{\sqrt{r(a,b)}}\right)^2\leq
  5\epsilon_1^2.\]

  Next, we consider the case that $a$ or $b$ is bad.  In this case, we
  will show that the weight of these configurations is so small that
  it overcomes the fact that $r(a,b)$ may not be close to 1.  If $r(a,b)\leq 1$ then
  $\pi(a,b)\left(\sqrt{r(a,b)}-\frac{1}{\sqrt{r(a,b)}}\right)^2\leq
  \pi(a)\pi(b)$.  Otherwise,
  $\pi(a,b)\left(\sqrt{r(a,b)}-\frac{1}{\sqrt{r(a,b)}}\right)^2\leq
  \frac{\pi(a,b)^2}{\pi(a)\pi(b)}.$  Either way,
 \[\pi(a,b)\left(\sqrt{r(a,b)}-\frac{1}{\sqrt{r(a,b)}}\right)^2\leq
 \pi(a)\pi(b) + \frac{\pi(a,b)^2}{\pi(a)\pi(b)}.\]
 
In order to upper bound $\frac{\pi(a,b)^2}{\pi(a)\pi(b)}$, we will use Lemma~\ref{lem:zBound} to first bound the ratio $\pi(a,b)/\pi(a)$ for \emph{good} $a$ and the ratio $\pi(a,b)/\pi(b)$ for all $b.$  Using Lemma~\ref{lem:zBound} part (3) we have for $a$ \emph{good}

$$\frac{\pi(a,b)}{\pi(a)} \leq \frac{w(a)w(b)w(a,b)\pi(a^*,b^*)}{w(a)\pi(a^*,b^*)\zhat'} \leq w(b)/\zhat'.$$
Using Lemma~\ref{lem:zBound} part (4) we have for all $b$
$$\frac{\pi(a,b)}{\pi(b)} \leq \frac{w(a)w(b)w(a,b)\pi(a^*,b^*)}{w(b)w(a^*,b)\pi(a^*,b^*)} \leq w(a).$$
 
  Next we will bound $\sum_{a,b} \frac{\pi(a,b)^2}{\pi(a)\pi(b)}$ for \emph{bad} $a$ and all $b$ using Lemma~\ref{lem:zBound} (part 1).

 \begin{eqnarray*}
 \sum_{a \text{ bad},b} \frac{\pi(a,b)^2}{\pi(a)\pi(b)} &=& \sum_{a \text{ bad}} \frac{1}{\pi(a)} \sum_b \frac{\pi(a,b)^2}{\pi(b)} \\
 &\leq& \sum_{a \text{ bad}} \frac{1}{\pi(a)}  \sum_b w(a)\pi(a,b)\\
 &= &  \sum_{a \text{ bad}} \frac{w(a)}{\pi(a)} \sum_b  \pi(a,b)\\
 &=& \sum_{a \text{ bad}} w(a) \leq \epsilon_1.
 \end{eqnarray*}
Next we will bound $\sum_{a,b} \frac{\pi(a,b)^2}{\pi(a)\pi(b)}$ for \emph{good} $a$ and \emph{bad} $b$ using Lemma~\ref{lem:zBound} (part 1).
  \begin{eqnarray*}
 \sum_{a \text{ good},b\text{ bad}} \frac{\pi(a,b)^2}{\pi(a)\pi(b)} &=& \sum_{b \text{ bad}}  \frac{1}{\pi(b)} \sum_{a \text{ good}} \frac{\pi(a,b)^2}{\pi(a)}\\
 &\leq&  \sum_{b \text{ bad}}  \frac{1}{\pi(b)} \sum_{a \text{ good}}w(b)\pi(a,b)/\zhat'\\
 &=&\sum_{b \text{ bad}}  \frac{ w(b)}{\pi(b)\zhat'} \sum_{a \text{ good}}\pi(a,b)\\
 &\leq& \sum_{b \text{ bad}} w(b) / \zhat' = \zhat''/\zhat'\\
 &\leq& \frac{\epsilon_1 \zhat}{(1-\epsilon_1)\zhat}= \frac{\epsilon_1}{1- \epsilon_1}.
 \end{eqnarray*}
 
\noindent Next, we will bound $\pi(a)\pi(b)$ for \emph{bad} $a$ and all $b$ using Lemma~\ref{lem:zBound} (part 1) and Equation~\ref{prodBound}
 \begin{eqnarray*}
 \sum_{a~\text{bad}, b} \pi(a)\pi(b) &\leq&  \sum_{a~\text{bad}, b} w(a)w(b)\pi(a^*,b^*)^2\ztilde \zhat \\
 &=& \ztilde \zhat\pi(a^*,b^*)^2  \sum_{a~\text{bad}} w(a) \sum_b w(b)\\
 &=&   \ztilde \zhat^2\pi(a^*,b^*)^2\ztilde''\\
 &\leq& \epsilon_1 (\ztilde \zhat\pi(a^*,b^*))(\zhat\pi(a^*,b^*))\\
 &\leq&\epsilon_1 (1-\epsilon_1)^{-3}.
 \end{eqnarray*}
 The last step uses the following bound on $\zhat\pi(a^*,b^*).$
 
 $$\zhat\pi(a^*,b^*) \leq \frac{\zhat'\ztilde'\pi(a^*,b^*)}{(1-\epsilon_1)}=\frac{1}{(1-\epsilon_1)}\sum_{a~\text{good},\atop b~\text{good}} w(a)w(b) \pi(a^*,b^*)  < \frac{1}{(1-\epsilon_1)}\sum_{(a,b)} \pi(a,b) =\frac{1}{(1-\epsilon_1)}.$$

 Finally, we will bound $\pi(a)\pi(b)$ for $b$ \emph{bad} and all $a$ using Lemma~\ref{lem:zBound} part (1) and Equation~\ref{prodBound}.\\
 \begin{eqnarray*}
 \sum_{a, b~\text{bad}} \pi(a)\pi(b) &\leq&  \sum_{a, b~\text{bad}} w(a)w(b)\pi(a^*,b^*)^2\ztilde \zhat \\
 &=& \ztilde \zhat\pi(a^*,b^*)^2  \sum_{b~\text{bad}} w(b) \sum_a w(a)\\
 &=&   \ztilde^2 \zhat\pi(a^*,b^*)^2\zhat''\\
 &\leq& \epsilon_1 (\ztilde \zhat\pi(a^*,b^*))^2\\
 &\leq&\epsilon_1 (1-\epsilon_1)^{-4}.
 \end{eqnarray*}
Putting this all together, we have
  \[\sum_{a, b}\pi(a,b)\left(\sqrt{r(a,b)}-\frac{1}{\sqrt{r(a,b)}}\right)^2\leq
  5\epsilon_1^2 + \epsilon_1 + \frac{\epsilon_1}{1-\epsilon_1} + \frac{\epsilon_1}{(1-\epsilon_1)^3}+  \frac{\epsilon_1}{(1-\epsilon_1)^4}\leq
  6\epsilon_1 = 1/n^2,\]
  as long as $\epsilon_1\leq .18$, which is true since $n\geq 2.$
Thus, we have shown $\Atilde$ and $\Ahat$ are $1/n$-orthogonal.


\subsection{The Markov chain $\mn$}\label{sec:mn}

Here we give the details to show how we can use Theorem~\ref{pprocess}, which gives a bound on the spectral gap of the particle process Markov chain $\mk$, to obtain a bound on the mixing time of the nearest-neighbor Markov chain $\mn.$

\vspace{.1in}
\noindent  {\bf The Nearest Neighbor Markov chain $\mn$ } 

\vspace{.05in}
\noindent {\tt Starting at any permutation $\sigma_0$, iterate the following:} 
\begin{itemize}
\item At time $t,$ choose a position $1< i\leq n$ uniformly at random.  
\item With probability $p_{\sigma_t(i), \sigma_t(i-1)}/2$, exchange the elements $\sigma_t(i)$ and $\sigma_t(i-1)$ to obtain $\sigma_{t+1}$.
\item Otherwise, do nothing so that $\sigma_{t+1} = \sigma_{t}.$
\end{itemize}

\noindent The chain $\mn$ connects the state space $\Omega$ and has the same stationary distribution as $\mk$ (see e.g.,~\cite{bmrs}).
We will use the bound on the
spectral gap of $\mk$ given by Theorem~\ref{pprocess} to prove the
following theorem. 

 \begin{theorem}\label{mnMixing} If the probabilities $\p$ are weakly monotonic and form a bounded \kP\, for $k\geq 2$, with $|C_i|\geq 2\psize$, then the mixing time $\tau_n$ of $\mn$ satisfies $\tau_n = O(n^{9}\ln(1/\epsilon)).$
\end{theorem}
\begin{proof}
Given our improved bound on the mixing time of $\mk$ from
Theorem~\ref{pprocess}, the remainder of this proof uses exactly the
same approach as~\cite{MS}, except that we use Corollary~\ref{cor:product}
to eliminate one factor of $n$.  We include a summary here for
completeness.  The complete details can be found in~\cite{MS}.
Instead of analyzing $\mn$ directly we will analyze an auxiliary chain
$\mtk$ that allows a larger set of transpositions (including those
allowed by $\mk$) and then use comparison techniques~\cite{dsc,RT98}
to obtain a bound for $\mn$.   

\vspace{.1in}
\noindent  {\bf The Transposition Markov chain $\mtk$} 

\vspace{.05in}
\noindent {\tt Starting at any permutation $\sigma_0$, iterate the following:} 
\begin{itemize}
\item At time $t,$ choose $1\leq i \leq n$ and $d \in \{L, R,N\}$ uniformly at random.  
\item If $d = L$, find the largest $j$ with $1 \leq j<i$ and $\C(\sigma_t(j)) \geq \C(\sigma_t(i))$ (if one exists).  If $\C(\sigma_t(j)) > \C(\sigma_t(i)),$ then 
with probability $1/2$,
exchange $\sigma_t(i)$ and $\sigma_t(j)$ to obtain $\sigma_{t+1}.$  
\item If $d = R$, find the smallest $j$ with $n\geq j>i$ and $\C(\sigma_t(j)) \geq \C(\sigma_t(i))$ (if one exists).  If $\C(\sigma_t(j)) > \C(\sigma_t(i)),$ then with probability 
$$\frac{1}{2}~ q_{\sigma_t(j),\sigma_t(i)}  \prod_{i<k<j}\left(q_{\sigma_t(j),\sigma_t(k)}q_{\sigma_t(k),\sigma_t(i)}\right),$$
exchange $\sigma_t(i)$ and $\sigma_t(j)$ to obtain $\sigma_{t+1}.$  
\item If $d = N,$  find the largest $j$ with $1\leq j<i$ and $\C(\sigma_t(j)) = \C(\sigma_t(i)).$  If such an element exists, then with probability $1/2$, exchange the elements $\sigma_t(i)$ and $\sigma_t(j)$ to obtain $\sigma_{t+1}.$ 

\item Otherwise, do nothing so that $\sigma_{t+1} = \sigma_t.$
\end{itemize}

The Markov chain $\mtk$ has the same
stationary distribution as $\mn$ and is a product of
$k+1$ independent Markov chains~\cite{MS}.  The first~$k$ chains involve moves
between particles in the same class and the $i$th is an unbiased
nearest-neighbor Markov chain over permutations on $c_i$ particles.
This chain has spectral gap $\Theta(c_i^{-3})$ (this is an 
unpublished result of Diaconis. See, e.g.~\cite{wilson}):
\begin{Lemma}[Diaconis]
The spectral gap of the unbiased nearest neighbor Markov chain
over permutations on $[n]$ is $(1-\cos(\pi/n))/(n-1)=\Theta(n^{-3})$.
\end{Lemma}
\noindent However, the chain $\mtk$ updates one of the first $k$ chains
only if direction $N$ is selected,  which happens with probability $c_i/(6n)$
for each particle class $1\leq i \leq k.$  Thus, the spectral gap of
the slowed-down version of this chain is $\Theta(1/(nc_i^2))$.
The final chain $\mk$ which we analyzed in Section~\ref{sec:perm}
allows only moves between different particle classes; it is updated when
direction $L$ or $R$ is selected (i.e. with probability $2/3$), so by
Theorem~\ref{pprocess}, it has spectral gap  $\Omega(n^{-2})$.
Therefore, by Corollary~\ref{cor:product}, the spectral gap of $\mtk$
is $\Omega(n^{-3})$.

The final step is to relate the mixing time of $\mn$ to that of $\mtk$
using the following comparison theorem~\cite{dsc}.  Let $P'$ and $P$ be two
reversible Markov chains on the same state space $\Omega$ with the
same stationary distribution $\pi$ and let $E(P) = \{(x,y): P(x,y) >
0\}$ and $E(P') = \{(x,y): P'(x,y) > 0\}$ denote the sets of edges of
the two graphs, viewed as directed graphs.  For each $x,y$ with
$P'(x,y)>0$, define a path $\gamma_{xy}$ using a sequence of states
$x=x_0,x_1,\cdots,x_k = y$ with $P(x_i,x_{i+1})>0$, and let
$|\gamma_{xy}|$ denote the length of the path.  Let $\Gamma(z,w) =
\{(x,y) \in E(P'): (z,w) \in \gamma_{xy}\}$ be the set of paths that
use the transition $(z,w)$ of $P$.  Finally, define  
\abovedisplayskip=3pt
\belowdisplayskip=3pt
$$A = \max_{(z,w) \in E(P)} \left \{\frac{1}{\pi(z)P(z,w)}\sum_{\Gamma(z,w)}|\gamma_{xy}|\pi(x)P'(x,y) \right \}.$$

\begin{theorem}[\cite{dsc}]\label{thm_comp}
With the above notation,  $\Gap(P)\geq \frac{1}{A} \Gap(P')$.
  \end{theorem}

\noindent Section~5 of~\cite{MS} proves the following lemma.

\begin{Lemma}[\cite{MS}]\label{lem_comp} If the probabilities~$\p$ are
  weakly monotonic and form a bounded \kP\, for $k\geq 2$, then $A=O(n^4)$.
\end{Lemma}

\noindent Combining Lemma~\ref{lem_comp} and Theorem~\ref{thm_comp}
with our bound on the gap of $\mtk$, we get $\Gap(\mn)=\Omega(n^{-7})$.
Finally, to get the mixing time of $\mn$, we let $q_* =
\max_{i<j} p_{i,j}/p_{j,i}$ then $\pi_* =  \min_{x\in \Omega}\pi(x)
\geq (q_*^{\binom{n}{2}}n!)^{-1}$ (see \cite{bmrs} and \cite{MS} for
more details), so $\log(1/\epsilon \pi_*) = O(n^2\ln \epsilon^{-1})$
since $q$ is bounded from above by a positive constant.  Applying
Theorem~\ref{gap}(a) we have that the
mixing time of $\mn$ is $O(n^9\ln\epsilon^{-1})$.  This proves
Theorem~\ref{mnMixing}.
\end{proof}

\section*{Acknowledgement}
We wish to thank Dana Randall for several useful discussions about the decomposition method.
\bibliographystyle{plain}
\bibliography{Decomposition}
\end{document}